\documentclass[copyright,creativecommons]{eptcs}
\usepackage{amsthm}
\usepackage{latex/packages}            

                         {\endgroup]\end{quote}}
\long\def\ignore#1{}

\newtheorem{theorem}{Theorem}
\newtheorem{lemma}{Lemma}[theorem]

\newcounter{substlemmacounter}
\newcommand*{\myfont}{\fontfamily{cmss}\selectfont}

\newcommand*{\ovA}[1]{%
  $\m@th\overline{\mbox{#1}\raisebox{2.3mm}{}}$%
}
\newcommand{\comment}[1]{}

\newcommand{\latestcommit}{fa651cd}

\newcommand{\clink}[3]{\href{https://github.com/LeaTrogni/formalizing-barbed-similarity-for-the-pi-calculus-in-beluga/blob/\latestcommit/code/#1\#L#2-L#3}{\includegraphics[height=0.02\textheight]{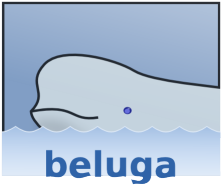}}}
\newcommand{\clinktwo}[1]{\href{https://github.com/LeaTrogni/formalizing-barbed-similarity-for-the-pi-calculus-in-beluga/blob/\latestcommit/code/#1}{\includegraphics[height=0.02\textheight]{latex/beluga-logo.png}}}

\newcommand{\bsim}{\mathbin{\overset{\bullet}{\leq}_B}}
\newcommand{\bpree}{\mathbin{\leq_{\mid s}}}
\newcommand{\bpre}{\mathbin{\leq_B}}   			
\newcommand*{\emacsfont}{\fontfamily{lmtt}\selectfont}

\definecolor{belugapurple}{RGB}{154,0,214}
\definecolor{belugared}{RGB}{225, 0, 0}
\definecolor{belugagreen}{RGB}{31,147,15}
\definecolor{belugablue}{RGB}{0,0,238}
\definecolor{belugapink}{RGB}{237,103,166}
\newcommand\delimone{{\color{belugapurple}\emacsfont\bfseries LF}\color{belugagreen}\aftergroup:}
\newcommand\delimtwo{{\color{belugapurple}\emacsfont\bfseries inductive}\color{belugagreen}\aftergroup:}
\newcommand\delimthree{{\color{belugapurple}\emacsfont\bfseries schema}\color{belugagreen}\aftergroup=}
\newcommand\delimfour{{\color{belugapurple}\emacsfont\bfseries rec}\color{belugablue}\aftergroup:}
\newcommand\delimfive{{/}\slshape\aftergroup/}
\newcommand\delimaux{{\color{belugapink}.}}
\newcommand\delimsix{\color{belugapink}--\aftergroup\delimaux}
\newcommand\delimseven{{\color{belugapurple}\emacsfont\bfseries coinductive}\color{belugagreen}\aftergroup:}

\lstdefinelanguage{Beluga}
{
  morekeywords=[1]{mlam,fn,case,of,total,in,type, impossible, let, and,ctype},
  keywordstyle=[1]\color{belugapurple}\emacsfont\bfseries,
  morecomment=[l]{\%},
  morecomment=[s]{\%\{}{\}\%},
  commentstyle=\color{belugared},
  sensitive=true,
}

\lstdefinestyle{belugastyle}{
 language=Beluga,
    basicstyle=\ttfamily\small,
    columns=flexible,
    keepspaces=true,
    showstringspaces=false,
    breaklines=true,
    breakatwhitespace=true,
    rangeprefix=\%\%\ ,
    rangesuffix=\ \%\%,
    includerangemarker=false,
    moredelim=**[is][\delimone]{LF}{:},
    moredelim=**[is][\delimtwo]{inductive}{:},
    moredelim=**[is][\delimthree]{schema}{=},
    moredelim=**[is][\delimfour]{rec}{:},
    moredelim=**[is][\delimfive]{/}{/},
    moredelim=**[is][\delimsix]{--}{.},
    moredelim=**[is][\delimseven]{coinductive}{:},
    literate={→}{{$\rightarrow$}}{1}
                  {⊢}{{$\vdash$}}{1}
                  {⇒}{{$\Rightarrow$}}{1}
}

\lstdefinestyle{belugastyleframes}{
 language=Beluga,
    basicstyle=\ttfamily\footnotesize,
    columns=flexible,
    keepspaces=true,
    showstringspaces=false,
    breaklines=true,
    breakatwhitespace=true,
    rangeprefix=\%\%\ ,
    rangesuffix=\ \%\%,
    includerangemarker=false,
    moredelim=**[is][\delimone]{LF}{:},
    moredelim=**[is][\delimtwo]{inductive}{:},
    moredelim=**[is][\delimthree]{schema}{=},
    moredelim=**[is][\delimfour]{rec}{:},
    moredelim=**[is][\delimfive]{/}{/},
    moredelim=**[is][\delimsix]{--}{.},
    moredelim=**[is][\delimseven]{coinductive}{:},
    frame=tblr,
    captionpos=b,
    literate={→}{{$\rightarrow$}}{1}
                  {⊢}{{$\vdash$}}{1}
                  {⇒}{{$\Rightarrow$}}{1}
}

\usepackage{tabularx}
\usepackage{array}
\usepackage{booktabs,wrapfig}

\counterwithin{theorem}{section}
\counterwithin{lemma}{section}

\title{Barbed Similarity for the $\pi$-Calculus in Beluga:\\ A Case Study in Coinductive Reasoning
}
\author{
  Lea Trogni
 \institute{Dipartimento di Matematica,\\
  Universit\`{a} degli Studi di Milano, Italy
}
\and
Gabriele Cecilia
\orcidlink{0009-0007-7797-5008}
\institute{School of Computer \& Cyber Sciences,\\
 Augusta University, Augusta, USA}
\and
Alberto Momigliano
\orcidlink{0000-0003-0942-4777}
 \institute{Dipartimento di Informatica,\\
  Universit\`{a} degli Studi di Milano, Italy}
}
\def\titlerunning{Barbed Similarity for the $\pi$-Calculus in Beluga}
\def\authorrunning{Trogni, Cecilia \& Momigliano  }

\begin{document}
\maketitle

\begin{abstract}
We formalize strong barbed similarity for the $\pi$-calculus in the Beluga proof assistant, completing a line of work addressing the Concurrent Calculi Formalization Benchmark. By extending previous developments to include replication, we give a coinductive encoding of behavioral equivalence based on barbs and internal actions. Using Beluga’s copattern-based coinduction, we obtain concise and compositional proofs, including compatibility properties and a context lemma characterizing barbed precongruence. The case study demonstrates the effectiveness of combining HOAS and coinductive reasoning for mechanizing concurrent calculi.
\end{abstract}

\section{Introduction}\label{sec:intro}

One of the reasons why people flock to  cinemas to watch sequels and ``rebooted'' movies
is that they know what they are getting into: they are familiar with
the characters and the basic setup. Keeping with the cinematographic
metaphor, we will cut to the chase and assume that the reader is
already on board with the idea that the best assurance of the
correctness of  properties of a formal calculus is a
machine-checked proof. In fact, this is becoming the norm in the
semantics of programming languages; it is less established when
\emph{concurrency} is concerned, although the latter is of
overwhelming importance in modern computing.

Enter the Concurrent Calculi Formalization Benchmark~\cite{ConcBench} (CCFB), which 
provides a suite of problems designed to evaluate the formalization of
concurrent and distributed language models, in particular process
calculi. Following the methodology of the POPLMark challenge, CCFB
aims to assess current mechanization techniques, identify optimal
formalization patterns, and ideally drive improvements in proof assistant
tooling.
The CCFB
framework 
 structures this  evaluation around three orthogonal challenges: the management of
linear resources, the handling of scope extrusion, and the
implementation of coinductive proof techniques. For the latter, the
challenge is to prove the ``context lemma'' for \emph{strong barbed
  bisimilarity}, namely that this notion of bisimilarity can be turned
into a congruence by making it sensitive to substitution and parallel
composition.


Induction is one of the great success stories of proof assistants:
although systems differ in the details of their foundations and
automation, inductive definitions and proofs by induction are supported
in essentially comparable ways across the board. Coinduction has had a
less uniform history; for a survey,  we refer to~\cite{HandPACoind}.
The earliest systematic approach, due to
Paulson, treats coinductive predicates as greatest fixed points
of monotone operators, from which the corresponding coinduction
principles are derived. 

A different tradition, familiar from Rocq/Coq, is based on guarded
corecursion: coinductive objects are introduced by constructors, and
recursive calls must occur under such constructors to ensure
productivity. This discipline is often brittle in proof developments,
especially when coinductive arguments are nested, combined with
induction, or hidden behind auxiliary lemmas. Agda has explored
several alternatives over time, from delays to sized types and
observational presentations of codata to guarded type theories with
the ``later'' modality~\cite{Nakano00}. Similarly, Beluga has embraced sized types and 
copatterns~\cite{ThibodeauCP16}. These developments are particularly
relevant for mechanized concurrency, where bisimulations and up-to
techniques routinely require coinductive proofs that are interleaved
with substantial inductive reasoning about syntax, transitions,
substitutions, and contexts.

The present paper is the conclusion of the Beluga ``trilogy'' of
solutions to CCFB, initiated with~\cite{ZackonSMP25} 
 --- featuring, among other contributions, a solution to the linearity challenge --- and
followed by~\cite{Cecilia24} with the mechanization of the Harmony Lemma. As in any good sequel,
we borrow from the latter paper the implementation of the syntax
and the LTS-based operational semantics of the $\pi$-calculus, while
crucially extending it with \emph{replication}. After all, this is what makes
behavioral equivalence interesting.  And interesting it turns out to be, as there is a twist in this movie: the rule governing replication as proposed in~\cite{ConcBench} is insufficient to establish the main result of the challenge. In fact, under this rule, structural congruence is not included in strong barbed similarity. This is problematic, because structural
congruence is the baseline syntactic equivalence between processes and
the context lemma relies crucially on this specific inclusion.

While this lapse is mildly embarrassing, given the singleton intersection between the authors of the Benchmark and of the present paper, it is once more a reaffirmation of the usefulness of mechanizations in proof assistants.

\smallskip
We will assume familiarity with the basic notions of
the $\pi$-calculus as in the first chapters of~\cite{SWBook}, as well
as a working knowledge of Beluga, both of its syntax and of its
approach to proof checking.  In particular we will only briefly touch
upon the way Beluga handles coinduction via observations and refer the reader to~\cite{ThibodeauCP16} for the theory and to~\cite{MomiglianoPT19} for an application.

In the following, the statements of informal lemmas, theorems and
proofs are hyperlinked via the accompanying cute icon {\includegraphics[height=0.02\textheight]{latex/beluga-logo.png}} to their formalization in the repository:
\begin{center}
  \begin{footnotesize}
    \url{https://github.com/LeaTrogni/formalizing-barbed-similarity-for-the-pi-calculus-in-beluga}
  \end{footnotesize}
\end{center}
As a final note,  we have cut one corner:  we have
concentrated on similarity, precongruence, etc., rather than bisimilarity and
congruence. Extending the results to the symmetric case would only duplicate the code with no new insights and could  be left to automation, perhaps a coding agent.

\section{The \texorpdfstring{$\pi$}{pi}-Calculus and its Operational Semantics}\label{sec:picalc}
To make the paper self-contained, we present
the main definitions of the syntax and the labelled transition system (LTS)
of the fragment of the $\pi$-calculus under study.
For more details about the informal definitions,
we refer the reader to~\cite{SWBook}.
\subsection{Syntax}
We follow the definitions of CCFB3 excluding sums and (mis)match, while
we diverge from it by allowing the calculus to be name-passing
rather than value-passing. That original separation of concerns is not
a simplification in the HOAS approach, but just a quirk. 
\begin{equation*}
  P, Q \ \vcentcolon \vcentcolon = \ \textbf{0} \ \mid \ x(y).P \ \mid \ \bar{x}y.P \ \mid \ (P \mid Q) \ \mid \ (\nu x) P \ \mid \ ! P
\end{equation*} 

Recall that the input prefix $x(y).P$ and the restriction $(\nu y) P$ both bind
the name $y$ in $P$, 
while any other
occurrence of names in a process is free. 
Accordingly, we write {\myfont fn($P$)} and {\myfont bn($P$)} for the
sets of free and bound names occurring in a process.

Beluga is based on a two-level system, with the LF level for representing data and a computation level for reasoning about it via recursion and  pattern matching. In the encoding, names are defined by
\begin{wrapfigure}[12]{r}{0.48\textwidth}
\vspace{-0.8\baselineskip}
  \begin{minipage}{\linewidth}
    \centering
\begin{lstlisting}[
      style=belugastyleframes,
      basicstyle=\footnotesize\ttfamily,
      linewidth=\linewidth,
      linerange={Names_and_Processes-End}
    ]
%%%%%%%%%%%%%%%%%%%%%%%%%%%%%%%%%%%%%%%%%%%%%%%%%%%%%%%%%

%%% Definitions %%%
% In order to encode input and restriction processes, which both bind names, we exploit higher order abstract syntax
% and use higher-order functions from names to proc.
% In this way we don't need to give an explicit name to bound names and deal with alpha-renaming or substitution.

%% Names_and_Processes %%
LF names: type =;

LF proc: type =
 | p_zero: proc                                
 | p_in: names → (names → proc) → proc         
 | p_out: names → names → proc → proc          
 | p_par: proc → proc → proc                   
 | p_res: (names → proc) → proc                
 | p_rep: proc → proc; ---infix p_par 11 left.

schema ctx = names;
%% End %%

% Equality of processes
LF eqp: proc → proc → type =
  | refp: eqp X X
;

% Equality of names
LF eqn: names → names → type =
  | refl: eqn N N
;

% empty type
LF not_possible : type = 
;



% Beniamino's trick - Predicate for contexts for a couple of terms
% Meta level
% PCtx [g' |- P] [g |- CP] [g' |- Q] [g |- CQ] holds if there is a context C such that C[P] = CP and C[Q] = CQ

%% Contexts %%
inductive PCtx: (g:ctx)(g':ctx)[g' ⊢ proc] → [g ⊢ proc] → [g' ⊢ proc] → [g ⊢ proc] → ctype =
  | pidM: PCtx [g ⊢ P] [g ⊢ P] [g ⊢ Q] [g ⊢ Q]
  | pinM: PCtx [g' ⊢ P] [g,x:names ⊢ CP] [g' ⊢ Q] [g,x:names ⊢ CQ] 
       → PCtx [g' ⊢ P] [g ⊢ p_in X (\x.CP)] [g' ⊢ Q] [g ⊢ p_in X (\x.CQ)]
  | poutM: PCtx [g' ⊢ P] [g ⊢ CP] [g' ⊢ Q] [g ⊢ CQ] 
        → PCtx [g' ⊢ P] [g ⊢ p_out X Y CP] [g' ⊢ Q] [g ⊢ p_out X Y CQ]
  ...								 
%% End %%

% Structural congruence
LF cong: proc → proc → type =
% Abelian Monoid Laws for Parallel Composition
  | par_assoc: cong (P p_par (Q p_par R)) ((P p_par Q) p_par R)
  | par_unit: cong (P p_par p_zero) P
  | par_comm: cong (P p_par Q) (Q p_par P)
% Scope Extension Laws
  | sc_ext_zero: cong (p_res \x.p_zero) p_zero
  | sc_ext_par: cong ((p_res P) p_par Q) (p_res \x.((P x) p_par Q))
  | sc_ext_res: cong (p_res \x.(p_res \y.(P x y))) (p_res \y.(p_res \x.(P x y)))
% Replication unfolding
  | rep_unfold: cong (p_rep P) (P p_par (p_rep P))
% Compatibility Laws
  | c_in: ({y:names} cong (P y) (Q y)) → cong (p_in X P) (p_in X Q)
  | c_out: cong P Q → cong (p_out X Y P) (p_out X Y Q)
  | c_par: cong P P' → cong (P p_par Q) (P' p_par Q)
  | c_res: ({x:names} cong (P x) (Q x)) → cong (p_res P) (p_res Q)
  | c_rep: cong P Q -> cong (p_rep P) (p_rep Q)
% Equivalence Relation Laws
  | c_ref: cong P P
  | c_sym: cong P Q → cong Q P
  | c_trans: cong P Q → cong Q R → cong P R
;
--infix cong 11 left.

%% Structural_congruence %%
inductive Cong: (g:ctx) [g ⊢ proc] → [g ⊢ proc] → ctype =
  ...
% Replication unfolding
  | Rep_unfold: Cong [g ⊢ (p_rep P)] [g ⊢ (P p_par (p_rep P))]
% Context closure
  | C_ctx: Cong [g' ⊢ P] [g' ⊢ Q] → PCtx [g' ⊢ P] [g ⊢ CP] [g' ⊢ Q] [g ⊢ CQ] 
        → Cong [g ⊢ CP] [g ⊢ CQ]
% Equivalence Relation Laws
  | C_ref: Cong [g ⊢ P] [g ⊢ P]
  | C_sym: Cong [g ⊢ P] [g ⊢ Q] → Cong [g ⊢ Q] [g ⊢ P]
  | C_trans: Cong [g ⊢ P] [g ⊢ Q] → Cong [g ⊢ Q] [g ⊢ R] → Cong [g ⊢ P] [g ⊢ R]
;
%% End %%

%%%%%%%%%%%%%%%%%%%%%%%%%%%%%%%%%%%%%%%%%%%%%%%%%%%%%%%%%

%%% LTS Semantics %%%
% We follow "Pi-Calculus in (Co)Inductive Type Theory" [Honsell et al. 2001] for the encoding of late LTS semantics.
% We define two different types for free/bound actions, and two different relations for transitions via free/bound actions.
% The result of a free transition is a process, while the result of a bound transition is a function from names to processes:
% instead of stating the bound name involved in the transition explicitly, that name is the argument of the aforementioned function.

% Free Actions
LF f_act: type =
  | f_tau: f_act
  | f_out: names → names → f_act
;

% Bound Actions
LF b_act: type =
  | b_in: names → b_act
  | b_out: names → b_act
;

% Equality of actions
LF eqf: f_act → f_act → type =
  | reff: eqf A A
;

LF eqb: b_act → b_act → type =
  | refb: eqb A A
;

% Transition Relation
LF fstep: proc → f_act → proc → type =
  | fs_out: fstep (p_out X Y P) (f_out X Y) P
  | fs_par1: fstep P A P' → fstep (P p_par Q) A (P' p_par Q)
  | fs_par2: fstep Q A Q' → fstep (P p_par Q) A (P p_par Q')
  | fs_com1: fstep P (f_out X Y) P' → bstep Q (b_in X) Q'
    → fstep (P p_par Q) f_tau (P' p_par (Q' Y))
  | fs_com2: bstep P (b_in X) P' → fstep Q (f_out X Y) Q'
    → fstep (P p_par Q) f_tau ((P' Y) p_par Q')
  | fs_res: ({z:names} fstep (P z) A (P' z))
    → fstep (p_res P) A (p_res P')
  | fs_close1: bstep P (b_out X) P' → bstep Q (b_in X) Q'
    → fstep (P p_par Q) f_tau (p_res \z.((P' z) p_par (Q' z)))
  | fs_close2: bstep P (b_in X) P' → bstep Q (b_out X) Q'
    → fstep (P p_par Q) f_tau (p_res \z.((P' z) p_par (Q' z)))
  | fs_rep: fstep P A P' → fstep (p_rep P) A (P' p_par (p_rep P))
  | fs_rep_com: fstep P (f_out X Y) P1 → bstep P (b_in X) P2
    → fstep (p_rep P) f_tau ((P1 p_par (P2 Y)) p_par (p_rep P))
  | fs_rep_close: bstep P (b_out X) P1 → bstep P (b_in X) P2
    → fstep (p_rep P) f_tau ((p_res \z.((P1 z) p_par (P2 z))) p_par (p_rep P))

and bstep: proc → b_act → (names → proc) → type =
  | bs_in: bstep (p_in X P) (b_in X) P
  | bs_par1: bstep P A P' → bstep (P p_par Q) A \x.((P' x) p_par Q)
  | bs_par2: bstep Q A Q' → bstep (P p_par Q) A \x.(P p_par (Q' x))
  | bs_res: ({z:names} bstep (P z) A (P' z))
    → bstep (p_res P) A \x.(p_res \z.(P' z x))
  | bs_open: ({z:names} fstep (P z) (f_out X z) (P' z))
    → bstep (p_res P) (b_out X) P'
  | bs_rep: bstep P A P' → bstep (p_rep P) A \x.((P' x) p_par (p_rep P))
;


%%% Barbed similarity %%%

inductive Barb_in : (g:ctx) [g |- proc] -> [g |- names] -> ctype =
  | Barb_in_b : [g |- bstep P (b_in X) (\x. P')] -> Barb_in [g |- P] [g |- X]
;

inductive Barb_out : (g:ctx) [g |- proc] -> [g |- names] -> ctype =
  | Barb_out_f : [g |- fstep P (f_out X Y) P'] -> Barb_out [g |- P] [g |- X]
  | Barb_out_b : [g |- bstep P (b_out X) (\x. P')] -> Barb_out [g |- P] [g |- X]
;

LF barb_in_rew : proc -> names -> type =
  | biw_zero : barb_in_rew ((p_in X Q) p_par R) X
  | biw_add: ({x : names} barb_in_rew (P x) X) -> barb_in_rew (p_res P) X
;

LF barb_in_red : proc -> names -> type =
  | barb_in_red_def: P cong P' -> barb_in_rew P' X -> barb_in_red P X
;

LF barb_out_rew : proc -> names -> type =
  | bow_zero : barb_out_rew ((p_out X Y Q) p_par R) X
  | bow_add: ({x : names} barb_out_rew (P x) X) -> barb_out_rew (p_res P) X
;

LF barb_out_red : proc -> names -> type =
  | barb_out_red_def: P cong P' -> barb_out_rew P' X -> barb_out_red P X
;

%% BarbSim %%
coinductive BarbSim : (g:ctx) [g ⊢ proc] → [g ⊢ proc] → ctype =
  | (BarbSim_barb_in : BarbSim [g ⊢ P] [g ⊢ Q])
              :: Barb_in [g ⊢ P] [g ⊢ X] → Barb_in [g ⊢ Q] [g ⊢ X]
  | (BarbSim_barb_out : BarbSim [g ⊢ P] [g ⊢ Q])
              :: Barb_out [g ⊢ P] [g ⊢ X] → Barb_out [g ⊢ Q] [g ⊢ X]
  | (BarbSim_tau : BarbSim [g ⊢ P] [g ⊢ Q])
              :: [g ⊢ fstep P f_tau P'] → Ex_sim_barb [g ⊢ P] [g ⊢ Q] [g ⊢ P']

and inductive Ex_sim_barb : (g:ctx) [g ⊢ proc] → [g ⊢ proc] → [g ⊢ proc] → ctype =
  | Ex_sim_barb_def : [g ⊢ fstep Q f_tau Q'] → BarbSim [g ⊢ P'] [g ⊢ Q'] 
                   → Ex_sim_barb [g ⊢ P] [g ⊢ Q] [g ⊢ P']
;
%% End %%

%%% Barbed similarity up to%%%

coinductive BarbSimUpTo : (g:ctx) [g |- proc] -> [g |- proc] -> ctype =
  | (BarbSimUpTo_barb_in : BarbSimUpTo [g |- P] [g |- Q])
              :: Barb_in [g |- P] [g |- X] -> Barb_in [g |- Q] [g |- X]
  | (BarbSimUpTo_barb_out : BarbSimUpTo [g |- P] [g |- Q])
              :: Barb_out [g |- P] [g |- X] -> Barb_out [g |- Q] [g |- X]
  | (BarbSimUpTo_f_tau : BarbSimUpTo [g |- S1] [g |- S2])
              :: [g |- fstep S1 f_tau S1'] -> Ex_sim_barb_up_to [g |- S1] [g |- S2] [g |- S1']

and inductive Ex_sim_barb_up_to : (g:ctx) [g |- proc] -> [g |- proc] -> [g |- proc] -> ctype =
  | Ex_sim_barb_up_to_def : [g |- fstep Q f_tau Q'] -> BarbSim [g |- P'] [g |- P''] -> BarbSimUpTo [g |- P''] [g |- Q''] -> BarbSim [g |- Q''] [g |- Q'] 
                 -> Ex_sim_barb_up_to [g |- P] [g |- Q] [g |- P']
;

%%% Barbed precongruence %%%
inductive BarbPre: (g':ctx) [g' |- proc] -> [g' |- proc] -> ctype =
  | BC:  (g':ctx) {P:[g' |- proc]} {Q:[g' |- proc]} 
    ((g : ctx) {CP:[g |- proc]} {CQ:[g |- proc]} PCtx [g' ⊢ P] [g ⊢ CP] [g' ⊢ Q] [g ⊢ CQ] -> BarbSim [g |- CP] [g |- CQ]) 
    -> BarbPre [g' |- P] [g' |- Q]
;

%%% Characterizations of barbed precongruence

% Context lemma characterization of barbed precongruence
inductive BarbPre': (g:ctx) [g |- proc] -> [g |- proc] -> ctype =
  | BC' : ((h:ctx) {$S : $[h |- g]} {R : [h |- proc]} -> BarbSim [h |- P[$S] p_par R] [h |- Q[$S] p_par R]) -> BarbPre' [g |- P] [g |- Q]
;

% Largest precongruence included in barbed similarity
coinductive LargPreCong: (g':ctx) [g' |- proc] -> [g' |- proc] -> ctype =
  | (LP_BarbSim : LargPreCong [g' |- P] [g' |- Q])
              :: BarbSim [g' |- P] [g' |- Q]
  | (LP_PreCong : LargPreCong [g' |- P] [g' |- Q])
              :: PCtx [g' ⊢ P] [g ⊢ CP] [g' ⊢ Q] [g ⊢ CQ] -> LargPreCong [g ⊢ CP] [g ⊢ CQ]
;


%%% OTHER IMPLICATION OF THE CONTEXT LEMMA

% Definition of natural numbers
LF nat : type =
  | z : nat
  | succ : nat -> nat
;

LF eqnat : nat -> nat -> type =
  | refnat: eqnat n n
;

LF leq : nat -> nat -> type =
  | samenat: leq n n
  | less: leq m n -> leq m (succ n)
;

% size of a context
inductive Dim : {g:ctx} [|- nat] -> ctype =
  | Dim_z: Dim [] [|- z]
  | Dim_succ: Dim [g] [|- n] -> Dim [g,x:names] [|- succ n]
;

% Multi_step_tau [|- n] [g |- P] [g |- Q] iff P evolves to Q in exactly n tau-transitions
inductive Multi_step_tau : (g:ctx) [|- nat] -> [g |- proc] -> [g |- proc] -> ctype =
  | Ms_tau_z: Multi_step_tau [|- z] [g |- P] [g |- P]
  | Ms_tau_succ: [g |- fstep P1 f_tau P2] -> Multi_step_tau [|- n] [g |- P2] [g |- P3] -> Multi_step_tau [|- succ n] [g |- P1] [g |- P3]
;

% prefix of a context
inductive Ctx_Pref: {g1:ctx} {g3:ctx} ctype =
  | Ctx_Pref_same: Ctx_Pref [g1] [g1]
  | Ctx_Pref_add: Ctx_Pref [g1] [g3] -> Ctx_Pref [g1] [g3,x:names]
;

% concatenation of contexts
inductive Ctx_App: {g1:ctx} {g2:ctx} {g3:ctx} ctype =
  | Ctx_App_empty: Ctx_App [g1] [] [g1]
  | Ctx_App_add: Ctx_App [g1] [g2] [g3] -> Ctx_App [g1] [g2,x:names] [g3,x:names]
;

% same name in a bigger context
inductive Name_Exp: (g2:ctx) (g3:ctx) [g2 |- names] -> [g3 |- names] -> ctype =
  | Name_Exp_same: Name_Exp [g2 |- X] [g2 |- X]
  | Name_Exp_add: Name_Exp [g2 |- X2] [g3 |- X3] -> Name_Exp [g2 |- X2] [g3,x:names |- X3[..]]
;

% concatenation of substitutions
inductive Subst_App: (g1:ctx) (g2:ctx) (g3:ctx) $[g1 |- g1] -> $[g2 |- g2] -> $[g3 |- g3] -> ctype =
  | Subst_App_empty: Subst_App $[g1 |- $S1] $[|- ^] $[g1 |- $S1]
  | Subst_App_add: Subst_App $[g1 |- $S1] $[g2 |- $S2] $[g3 |- $S3] 
    -> Name_Exp [g2,x:names |- X2] [g3,x:names |- X3] 
    -> Subst_App $[g1 |- $S1] $[g2,x:names |- $S2[..],X2] $[g3,x:names |- $S3[..],X3]
;

% expansion of a substitution
inductive Subst_Exp: (g1:ctx) (g3:ctx) $[g1 |- g1] -> $[g3 |- g3] -> ctype =
  | Subst_Exp_Def: Subst_App $[g1 |- $S1] $[g2 |- ..] $[g3 |- $S3] -> Subst_Exp $[g1 |- $S1] $[g3 |- $S3]
;

% Context built from substitution
%{inductive Ctx_pref_zero: (g':ctx) {g:ctx} $[g' |- g'] -> [g',x:names |- proc] -> ctype =
  | Ctx_pref_zero_empty: Ctx_pref_zero [] $[g' |- ..] [g',x:names |- p_zero]
  | Ctx_pref_zero_add: Partial_subst [g] [g'] Ctx_pref_zero [g] $[g' |- $S] [g',x:names |- P0] -> Ctx_pref_zero [g,x:names] $[g |- $S,Y] [g,x:names |- p_out x Y[..] P0]
;

inductive Ctx_pref_proc: (g:ctx) {g':ctx} [g |- proc] -> [g,x:names |- proc] -> ctype =
  | Ctx_pref_proc_empty: Ctx_pref_proc [] [g |- P] [g,x:names |- P[..]]
  | Ctx_pref_proc_add: Ctx_pref_proc [g'] [g |- P] [g,x:names |- P'] -> Ctx_pref_proc [g',x:names] [g |- P] [g,x:names |- p_in x \x. P'[..]]
;}%

%{inductive Ctx_par_subst: (g':ctx) {g:ctx} $[g |- g] -> [g' |- proc] -> [g' |- proc] -> [g',x:names |- proc] -> ctype =
  | Ctx_par_subst_empty: Ctx_par_subst [] $[ |- ^] [g' |- R] [g' |- P] [g',x:names |- (p_zero p_par P[..]) p_par R[..]]
  | Ctx_par_subst_add: Ctx_par_subst [g] $[g |- $S] [g',x:names |- R] [g',x:names |- P] [g',x:names,y:names |- (Po p_par Pi) p_par R[..]] -> Name_Exp [g,x:names |- X] [g',x:names |- X'] -> Ctx_par_subst [g,x:names] $[g,x:names |- $S[..],X] [g',x:names |- R] [g',x:names |- P] [g',x:names,y:names |- ((p_out x X'[..] Po) p_par (p_in y \x. Pi))]
;}%
\end{lstlisting}
    \vspace{-0.5\baselineskip}
    \caption{Encoding of names and processes.}
    \label{fig:proc}
    \vspace{-1.2\baselineskip}
  \end{minipage}
\end{wrapfigure}
   an LF type {\ttfamily \small
  \color{belugagreen}names} without any constructor, together with a
context schema that allows names to occur in other (open) LF
objects. Since we define no other context schema, all the
variables in open terms will have type {\ttfamily \small
  \color{belugagreen}names}.
Processes are LF objects, encoded using (weak) HOAS for input and restriction, as well known and detailed in~\cite{Cecilia24};
their encoding is displayed in Fig.~\ref{fig:proc}.
Throughout the development, lowercase identifiers denote LF types, while those beginning with an uppercase letter denote computation-level data.

\subsection{Operational Semantics}

\emph{Actions} include inputs $x(y)$, free and bound outputs $\bar{x}y$ and $\bar{x}(y)$, 
and internal communications $\tau$~\clink{1\_definitions.bel}{104}{114}. In inputs and bound outputs, the occurrences of 
the name $y$ are bound, while all other name occurrences in actions are free; 
from this, we obtain the standard notions of
free names, bound names and names occurring in an action $\alpha$, respectively denoted as {\myfont
fn($\alpha$)}, {\myfont bn($\alpha$)} and {\myfont n($\alpha$)}.

%

Our LTS, displayed in Fig.~\ref{fig:lts}, presents some differences from the one
in CCFB3. First, and less importantly, we have adopted \emph{late}
instead of \emph{early} semantics: not only can they  be proven  equivalent in terms of transitions, see~\cite{Cecilia24} for the mechanization of this folk result, but
they induce the same strong barbed similarity, as established by Lemma~\ref{le:barb_equiv} below and the discussion that follows.
Secondly, as we teased in the introduction, we add two replication rules for
communication to the LTS presented in CCFB3, since without them structural
congruence is not included in strong barbed similarity: in the accompanying formalization~\clinktwo{replication_rule_change} 
we exhibit a pair of structurally congruent processes that are not barbed similar.
This conflicts with the intended role of structural congruence as the strongest equivalence between processes, and the proof of the context lemma relies on this specific inclusion.
The adoption of the two replication rules for communication
is not arbitrary; after all, our LTS follows that in Sangiorgi and Walker's textbook,
modulo late semantics. A popular alternative, found e.g.\ in~\cite{DBLP:journals/tcs/HonsellMS01,Bengtson2009},
uses a single rule to generate arbitrary copies of the replicated process.
We refer to Sangiorgi and Walker for a discussion of the trade-offs between these formulations; the correspondence is understood modulo structural congruence.


\begin{figure}[ht]
\centering
\begin{footnotesize}
\fbox{\vbox {\advance \hsize by -2\fboxsep \advance \hsize by -2\fboxrule \linewidth\hsize
\begin{mathpar}
 \inferrule[$\mkern 30mu$ S-In] { } {x(z).P \, \xrightarrow{x(z)} \, P}
  \and \mkern 75mu \inferrule[$\mkern 9mu$ S-Out] { } {\bar{x}y.P \, \xrightarrow{\bar{x}y} \, P} \mkern 16mu \\
  \inferrule[$\mkern 77mu$ S-Par-L] {P \xrightarrow{\alpha} P' \\ ${\myfont bn}$(\alpha)$ $\cap$ {\myfont fn}$(Q) = \emptyset} {P \mid Q \ \xrightarrow{\alpha} \ P' \mid Q}
  \mkern 68mu \inferrule[$\mkern 77mu$ S-Par-R] {Q \xrightarrow{\alpha} Q' \\ ${\myfont bn}$(\alpha)$ $\cap$ {\myfont fn}$(P) = \emptyset} {P \mid Q \ \xrightarrow{\alpha} \ P \mid Q'} \\
  \inferrule[$\mkern 44mu$ S-Com-L] {P \xrightarrow{\bar{x}y} P' \\ Q \xrightarrow{x(z)} Q'} {P \mid Q \, \xrightarrow{\tau} \, P' \mid Q'\{  y/z \}} \mkern 13mu
  \and \inferrule[$\mkern 44mu$ S-Com-R] {P \xrightarrow{x(z)} P' \\ Q \xrightarrow{\bar{x}y} Q'} {P \mid Q \, \xrightarrow{\tau} \, P' \{ y/z \} \mid Q'} \\
  \inferrule[$\mkern 50mu$ S-Res] {P \xrightarrow{\alpha} P' \\ z \notin$ {\myfont n}$(\alpha)} {(\nu z) P \, \xrightarrow{\alpha} \, (\nu z) P'}
  \and \mkern 22mu \mkern 8mu \inferrule[$\mkern 32mu$ S-Open] {P \xrightarrow{\bar{x}z} P' \\ z \neq x} {(\nu z)P \, \xrightarrow{\bar{x}(z)} \, P'} \mkern 12mu \\
  \inferrule[$\mkern 38mu$ S-Close-L] {P \xrightarrow{\bar{x}(z)} P' \\ Q \xrightarrow{x(z)} Q'} {P \mid Q \ \xrightarrow{\tau} \ (\nu z) (P' \mid Q')}
  \and \inferrule[$\mkern 38mu$ S-Close-R] {P \xrightarrow{x(z)} P' \\ Q \xrightarrow{\bar{x}(z)} Q'} {P \mid Q \ \xrightarrow{\tau} \ (\nu z) (P' \mid Q')}\\
  \inferrule[] { } { }
  \mkern -40mu
  \and
   \inferrule[$\mkern 17mu$ S-Rep] {P \xrightarrow{\alpha} P'} {! P \, \xrightarrow{\alpha} \, P' \mid ! P} \mkern 35mu
  \and \inferrule[$\mkern 32mu$ S-Rep-Com] {P \xrightarrow{\bar{x}y} P' \\ P \xrightarrow{x(z)} P''} {!P \, \xrightarrow{\tau} \, (P' \mid P''\{  y/z \}) \mid ! P} \mkern -8mu
  \and \inferrule[$\mkern 30mu$ S-Rep-Close] {P \xrightarrow{\bar{x}(z)} P' \\ P \xrightarrow{x(z)} P''} {! P \ \xrightarrow{\tau} \ ((\nu z) (P' \mid P'')) \mid ! P}
\end{mathpar}
}}
\end{footnotesize}
\vspace{-5mm}
\caption{Transition rules.}
\label{fig:lts}
\end{figure}

Following the previous formalizations in~\cite{MillerPI,DBLP:journals/tcs/HonsellMS01,DBLP:journals/tocl/TiuM10}, in Beluga we separate free steps, where the action does not bind any variable, from bound steps, where the action has a bound variable.
Since some rules involve both free and bound actions at the same time, the LF types for free and bound steps are mutually defined. This approach does duplicate some rules, but gets rid of \emph{all} provisos on free and bound occurrences, which are now taken care of by dependencies (or lack thereof) in second-order process variables.~\clink{1\_definitions.bel}{125}{155}

\subsection{Process Contexts and Structural Congruence}

The notion of \emph{(pre)congruence} is of paramount importance in the
$\pi$-calculus: both in the reduction semantics, through \emph{structural} congruence, and in the study of behavioral equivalence.  Informally, a congruence is an equivalence relation that
preserves the behavior of the constructors. Often, this is described
 in textbooks
(e.g.~\cite{SWBook}) via the notion of \emph{process
  context}: a process where a (non-degenerate) occurrence of 
\textbf{0} is replaced by a
hole. 
Given a context $C$ and a process $P$, $C[P]$ denotes the process
obtained by replacing the hole in $C$ with $P$.  Then, a process
\emph{pre}congruence is a binary relation $\mathcal{R}$ which is a preorder
such that if $P \mathbin{\mathcal{R}} Q$, then, for each context $C$,
$C[P] \mathbin{\mathcal{R}} C[Q]$.

 There
 is however a peculiarity: since the goal of contexts is to observe how
 processes behave in every possible environment, they need to be able
 to capture the free variables in the processes they are
 filled with. In other words,  contexts do not respect  $\alpha$-equivalence.

 This brings in some additional issues with respect to a mechanization, in
 particular in a HOAS setting, where such entities are a
 non-starter~\footnote{HOAS is indeed compatible with weaker notions
   such as \emph{evaluation} contexts, i.e., those that do not cross a
   binder.}. A workaround is to close the relation under
 \emph{compatibility} (\cite{Lassen98,Gay_Vasconcelos_2025}), as
 described by the rules in our case at the top of
 Fig.~\ref{fig:compcong}. Concretely, when we say that
 \emph{structural congruence} is the smallest congruence satisfying
 the axioms in the bottom of Fig.~\ref{fig:compcong}, we are
 instantiating $\cal R$ with $\equiv$.

This is indeed the way structural congruence was formalized as an LF type {\ttfamily \small \color{belugagreen}cong} in~\cite{Cecilia24}, which we have extended by adding the compatibility and unfolding laws for replication.~\clink{1\_definitions.bel}{51}{73}



\begin{figure}[th]
  \begin{footnotesize}
\centering
\fbox{\vbox {\advance \hsize by -2\fboxsep \advance \hsize by -2\fboxrule \linewidth\hsize
\begin{mathpar}
  \inferrule[$\mkern 38mu$ C-In] {P \mathbin{\mathcal{R}} Q} {x(y).P \, \mathbin{\mathcal{R}} \, x(y).Q}
  \and \inferrule[$\mkern 17mu$ C-Out] {P \mathbin{\mathcal{R}} Q} {\bar{x}y.P \, \mathbin{\mathcal{R}} \, \bar{x}y.Q}
  \and \inferrule[$\mkern 36mu$ C-Par] {P \mathbin{\mathcal{R}} P' \\ Q \mathbin{\mathcal{R}} Q'} {P \mid Q \, \mathbin{\mathcal{R}} \, P' \mid Q'}\\
  \inferrule[$\mkern 31mu$ C-Res] {P \mathbin{\mathcal{R}} Q} {(\nu x) P \, \mathbin{\mathcal{R}} \, (\nu x) Q}
  \and \inferrule[$\mkern 1mu$ C-Rep] {P \mathbin{\mathcal{R}} Q} {! P \, \mathbin{\mathcal{R}} \, ! Q}\\ 
   \hbox to 16cm{\leaders\hbox to 10pt{---}\hfil} \\
  \inferrule[$\mkern 40mu$ Par-Assoc] { } {P \mid (Q \mid R) \, \equiv \, (P \mid Q) \mid R}
  \mkern 92mu \inferrule[Par-Unit] { } {P \mid \textbf{0} \, \equiv \, P}
  \mkern 32mu \and \mkern 8mu \inferrule[$\mkern 2mu$ Par-Comm] { } {P \mid Q \, \equiv \, Q \mid P} \mkern 42mu \\
  \inferrule[] { } { }
  \mkern 35mu \inferrule[Sc-Ext-Zero] { } {\ (\nu x)\, \textbf{0} \, \equiv \, \textbf{0} \ \ }
  \mkern 76mu \inferrule[$\mkern 38mu$ Sc-Ext-Par] {x \notin$ {\myfont fn}$(Q)} {(\nu x) P \mid Q \, \equiv \, (\nu x) (P \mid Q)}
  \mkern 29mu \inferrule[\qquad \ \ Sc-Ext-Res] { } {(\nu x) (\nu y) P \, \equiv \, (\nu y) (\nu x) P} \\
  \inferrule[$\mkern -15mu$ Rep-Unfold] { } {! P \, \equiv \, P \mid ! P} \\
\end{mathpar}
}}
\end{footnotesize}
\vspace{-5mm}
\caption{Compatibility rules and structural congruence axioms.}
\label{fig:compcong}
\end{figure}

There are cases where going the compatibility route is inconvenient, if
not plain inadequate, one being the definition of \emph{strong barbed
  (pre)congruence}, as we shall see later on, and where contexts must
be encoded. In a breakthrough, the authors of~\cite{LancelotAV25}
introduced a clever trick to make contexts compatible (no pun
intended) with a HOAS encoding, by means of a quaternary relation that
implicitly captures the essence of a context:


\begin{equation*}
  \{(P,P',Q,Q') \mid P'=C[P], \ Q'=C[Q] \ \text{for some context} \ C\}
\end{equation*}

Unlike in the Abella setting where this relation was introduced,
a Beluga encoding is rather subtle as the LF contexts
must be made explicit. 
In fact, in general $\myfont\text{fn}(C[P]) \neq \myfont\text{fn}(P)$,
since the binders in $C$ may capture names that occur free in $P$,
and $C$ may also introduce names that are fresh for $P$;
while the latter can be avoided by working in an LF context that already contains
all free variables of $P$ and $C$, the former inevitably leads to
a mismatch between the two LF contexts.
Therefore, the formalization of this relation cannot be
an LF predicate, where all the related terms are defined in an implicit
context. Instead, we define the inductive type {\ttfamily \small
  \color{belugagreen}PCtx} in Fig.~\ref{fig:PCtx}, where the context of the second and
fourth processes is a prefix of the context of the first and third
processes.~\clink{1\_definitions.bel}{40}{48}

\begin{figure}[ht]
\begin{lstlisting}[style=belugastyleframes,linerange={Contexts-End},basicstyle=\footnotesize\ttfamily]
%%%%%%%%%%%%%%%%%%%%%%%%%%%%%%%%%%%%%%%%%%%%%%%%%%%%%%%%%

%%% Definitions %%%
% In order to encode input and restriction processes, which both bind names, we exploit higher order abstract syntax
% and use higher-order functions from names to proc.
% In this way we don't need to give an explicit name to bound names and deal with alpha-renaming or substitution.

%% Names_and_Processes %%
LF names: type =;

LF proc: type =
 | p_zero: proc                                
 | p_in: names → (names → proc) → proc         
 | p_out: names → names → proc → proc          
 | p_par: proc → proc → proc                   
 | p_res: (names → proc) → proc                
 | p_rep: proc → proc; ---infix p_par 11 left.

schema ctx = names;
%% End %%

% Equality of processes
LF eqp: proc → proc → type =
  | refp: eqp X X
;

% Equality of names
LF eqn: names → names → type =
  | refl: eqn N N
;

% empty type
LF not_possible : type = 
;



% Beniamino's trick - Predicate for contexts for a couple of terms
% Meta level
% PCtx [g' |- P] [g |- CP] [g' |- Q] [g |- CQ] holds if there is a context C such that C[P] = CP and C[Q] = CQ

%% Contexts %%
inductive PCtx: (g:ctx)(g':ctx)[g' ⊢ proc] → [g ⊢ proc] → [g' ⊢ proc] → [g ⊢ proc] → ctype =
  | pidM: PCtx [g ⊢ P] [g ⊢ P] [g ⊢ Q] [g ⊢ Q]
  | pinM: PCtx [g' ⊢ P] [g,x:names ⊢ CP] [g' ⊢ Q] [g,x:names ⊢ CQ] 
       → PCtx [g' ⊢ P] [g ⊢ p_in X (\x.CP)] [g' ⊢ Q] [g ⊢ p_in X (\x.CQ)]
  | poutM: PCtx [g' ⊢ P] [g ⊢ CP] [g' ⊢ Q] [g ⊢ CQ] 
        → PCtx [g' ⊢ P] [g ⊢ p_out X Y CP] [g' ⊢ Q] [g ⊢ p_out X Y CQ]
  ...								 
%% End %%

% Structural congruence
LF cong: proc → proc → type =
% Abelian Monoid Laws for Parallel Composition
  | par_assoc: cong (P p_par (Q p_par R)) ((P p_par Q) p_par R)
  | par_unit: cong (P p_par p_zero) P
  | par_comm: cong (P p_par Q) (Q p_par P)
% Scope Extension Laws
  | sc_ext_zero: cong (p_res \x.p_zero) p_zero
  | sc_ext_par: cong ((p_res P) p_par Q) (p_res \x.((P x) p_par Q))
  | sc_ext_res: cong (p_res \x.(p_res \y.(P x y))) (p_res \y.(p_res \x.(P x y)))
% Replication unfolding
  | rep_unfold: cong (p_rep P) (P p_par (p_rep P))
% Compatibility Laws
  | c_in: ({y:names} cong (P y) (Q y)) → cong (p_in X P) (p_in X Q)
  | c_out: cong P Q → cong (p_out X Y P) (p_out X Y Q)
  | c_par: cong P P' → cong (P p_par Q) (P' p_par Q)
  | c_res: ({x:names} cong (P x) (Q x)) → cong (p_res P) (p_res Q)
  | c_rep: cong P Q -> cong (p_rep P) (p_rep Q)
% Equivalence Relation Laws
  | c_ref: cong P P
  | c_sym: cong P Q → cong Q P
  | c_trans: cong P Q → cong Q R → cong P R
;
--infix cong 11 left.

%% Structural_congruence %%
inductive Cong: (g:ctx) [g ⊢ proc] → [g ⊢ proc] → ctype =
  ...
% Replication unfolding
  | Rep_unfold: Cong [g ⊢ (p_rep P)] [g ⊢ (P p_par (p_rep P))]
% Context closure
  | C_ctx: Cong [g' ⊢ P] [g' ⊢ Q] → PCtx [g' ⊢ P] [g ⊢ CP] [g' ⊢ Q] [g ⊢ CQ] 
        → Cong [g ⊢ CP] [g ⊢ CQ]
% Equivalence Relation Laws
  | C_ref: Cong [g ⊢ P] [g ⊢ P]
  | C_sym: Cong [g ⊢ P] [g ⊢ Q] → Cong [g ⊢ Q] [g ⊢ P]
  | C_trans: Cong [g ⊢ P] [g ⊢ Q] → Cong [g ⊢ Q] [g ⊢ R] → Cong [g ⊢ P] [g ⊢ R]
;
%% End %%

%%%%%%%%%%%%%%%%%%%%%%%%%%%%%%%%%%%%%%%%%%%%%%%%%%%%%%%%%

%%% LTS Semantics %%%
% We follow "Pi-Calculus in (Co)Inductive Type Theory" [Honsell et al. 2001] for the encoding of late LTS semantics.
% We define two different types for free/bound actions, and two different relations for transitions via free/bound actions.
% The result of a free transition is a process, while the result of a bound transition is a function from names to processes:
% instead of stating the bound name involved in the transition explicitly, that name is the argument of the aforementioned function.

% Free Actions
LF f_act: type =
  | f_tau: f_act
  | f_out: names → names → f_act
;

% Bound Actions
LF b_act: type =
  | b_in: names → b_act
  | b_out: names → b_act
;

% Equality of actions
LF eqf: f_act → f_act → type =
  | reff: eqf A A
;

LF eqb: b_act → b_act → type =
  | refb: eqb A A
;

% Transition Relation
LF fstep: proc → f_act → proc → type =
  | fs_out: fstep (p_out X Y P) (f_out X Y) P
  | fs_par1: fstep P A P' → fstep (P p_par Q) A (P' p_par Q)
  | fs_par2: fstep Q A Q' → fstep (P p_par Q) A (P p_par Q')
  | fs_com1: fstep P (f_out X Y) P' → bstep Q (b_in X) Q'
    → fstep (P p_par Q) f_tau (P' p_par (Q' Y))
  | fs_com2: bstep P (b_in X) P' → fstep Q (f_out X Y) Q'
    → fstep (P p_par Q) f_tau ((P' Y) p_par Q')
  | fs_res: ({z:names} fstep (P z) A (P' z))
    → fstep (p_res P) A (p_res P')
  | fs_close1: bstep P (b_out X) P' → bstep Q (b_in X) Q'
    → fstep (P p_par Q) f_tau (p_res \z.((P' z) p_par (Q' z)))
  | fs_close2: bstep P (b_in X) P' → bstep Q (b_out X) Q'
    → fstep (P p_par Q) f_tau (p_res \z.((P' z) p_par (Q' z)))
  | fs_rep: fstep P A P' → fstep (p_rep P) A (P' p_par (p_rep P))
  | fs_rep_com: fstep P (f_out X Y) P1 → bstep P (b_in X) P2
    → fstep (p_rep P) f_tau ((P1 p_par (P2 Y)) p_par (p_rep P))
  | fs_rep_close: bstep P (b_out X) P1 → bstep P (b_in X) P2
    → fstep (p_rep P) f_tau ((p_res \z.((P1 z) p_par (P2 z))) p_par (p_rep P))

and bstep: proc → b_act → (names → proc) → type =
  | bs_in: bstep (p_in X P) (b_in X) P
  | bs_par1: bstep P A P' → bstep (P p_par Q) A \x.((P' x) p_par Q)
  | bs_par2: bstep Q A Q' → bstep (P p_par Q) A \x.(P p_par (Q' x))
  | bs_res: ({z:names} bstep (P z) A (P' z))
    → bstep (p_res P) A \x.(p_res \z.(P' z x))
  | bs_open: ({z:names} fstep (P z) (f_out X z) (P' z))
    → bstep (p_res P) (b_out X) P'
  | bs_rep: bstep P A P' → bstep (p_rep P) A \x.((P' x) p_par (p_rep P))
;


%%% Barbed similarity %%%

inductive Barb_in : (g:ctx) [g |- proc] -> [g |- names] -> ctype =
  | Barb_in_b : [g |- bstep P (b_in X) (\x. P')] -> Barb_in [g |- P] [g |- X]
;

inductive Barb_out : (g:ctx) [g |- proc] -> [g |- names] -> ctype =
  | Barb_out_f : [g |- fstep P (f_out X Y) P'] -> Barb_out [g |- P] [g |- X]
  | Barb_out_b : [g |- bstep P (b_out X) (\x. P')] -> Barb_out [g |- P] [g |- X]
;

LF barb_in_rew : proc -> names -> type =
  | biw_zero : barb_in_rew ((p_in X Q) p_par R) X
  | biw_add: ({x : names} barb_in_rew (P x) X) -> barb_in_rew (p_res P) X
;

LF barb_in_red : proc -> names -> type =
  | barb_in_red_def: P cong P' -> barb_in_rew P' X -> barb_in_red P X
;

LF barb_out_rew : proc -> names -> type =
  | bow_zero : barb_out_rew ((p_out X Y Q) p_par R) X
  | bow_add: ({x : names} barb_out_rew (P x) X) -> barb_out_rew (p_res P) X
;

LF barb_out_red : proc -> names -> type =
  | barb_out_red_def: P cong P' -> barb_out_rew P' X -> barb_out_red P X
;

%% BarbSim %%
coinductive BarbSim : (g:ctx) [g ⊢ proc] → [g ⊢ proc] → ctype =
  | (BarbSim_barb_in : BarbSim [g ⊢ P] [g ⊢ Q])
              :: Barb_in [g ⊢ P] [g ⊢ X] → Barb_in [g ⊢ Q] [g ⊢ X]
  | (BarbSim_barb_out : BarbSim [g ⊢ P] [g ⊢ Q])
              :: Barb_out [g ⊢ P] [g ⊢ X] → Barb_out [g ⊢ Q] [g ⊢ X]
  | (BarbSim_tau : BarbSim [g ⊢ P] [g ⊢ Q])
              :: [g ⊢ fstep P f_tau P'] → Ex_sim_barb [g ⊢ P] [g ⊢ Q] [g ⊢ P']

and inductive Ex_sim_barb : (g:ctx) [g ⊢ proc] → [g ⊢ proc] → [g ⊢ proc] → ctype =
  | Ex_sim_barb_def : [g ⊢ fstep Q f_tau Q'] → BarbSim [g ⊢ P'] [g ⊢ Q'] 
                   → Ex_sim_barb [g ⊢ P] [g ⊢ Q] [g ⊢ P']
;
%% End %%

%%% Barbed similarity up to%%%

coinductive BarbSimUpTo : (g:ctx) [g |- proc] -> [g |- proc] -> ctype =
  | (BarbSimUpTo_barb_in : BarbSimUpTo [g |- P] [g |- Q])
              :: Barb_in [g |- P] [g |- X] -> Barb_in [g |- Q] [g |- X]
  | (BarbSimUpTo_barb_out : BarbSimUpTo [g |- P] [g |- Q])
              :: Barb_out [g |- P] [g |- X] -> Barb_out [g |- Q] [g |- X]
  | (BarbSimUpTo_f_tau : BarbSimUpTo [g |- S1] [g |- S2])
              :: [g |- fstep S1 f_tau S1'] -> Ex_sim_barb_up_to [g |- S1] [g |- S2] [g |- S1']

and inductive Ex_sim_barb_up_to : (g:ctx) [g |- proc] -> [g |- proc] -> [g |- proc] -> ctype =
  | Ex_sim_barb_up_to_def : [g |- fstep Q f_tau Q'] -> BarbSim [g |- P'] [g |- P''] -> BarbSimUpTo [g |- P''] [g |- Q''] -> BarbSim [g |- Q''] [g |- Q'] 
                 -> Ex_sim_barb_up_to [g |- P] [g |- Q] [g |- P']
;

%%% Barbed precongruence %%%
inductive BarbPre: (g':ctx) [g' |- proc] -> [g' |- proc] -> ctype =
  | BC:  (g':ctx) {P:[g' |- proc]} {Q:[g' |- proc]} 
    ((g : ctx) {CP:[g |- proc]} {CQ:[g |- proc]} PCtx [g' ⊢ P] [g ⊢ CP] [g' ⊢ Q] [g ⊢ CQ] -> BarbSim [g |- CP] [g |- CQ]) 
    -> BarbPre [g' |- P] [g' |- Q]
;

%%% Characterizations of barbed precongruence

% Context lemma characterization of barbed precongruence
inductive BarbPre': (g:ctx) [g |- proc] -> [g |- proc] -> ctype =
  | BC' : ((h:ctx) {$S : $[h |- g]} {R : [h |- proc]} -> BarbSim [h |- P[$S] p_par R] [h |- Q[$S] p_par R]) -> BarbPre' [g |- P] [g |- Q]
;

% Largest precongruence included in barbed similarity
coinductive LargPreCong: (g':ctx) [g' |- proc] -> [g' |- proc] -> ctype =
  | (LP_BarbSim : LargPreCong [g' |- P] [g' |- Q])
              :: BarbSim [g' |- P] [g' |- Q]
  | (LP_PreCong : LargPreCong [g' |- P] [g' |- Q])
              :: PCtx [g' ⊢ P] [g ⊢ CP] [g' ⊢ Q] [g ⊢ CQ] -> LargPreCong [g ⊢ CP] [g ⊢ CQ]
;


%%% OTHER IMPLICATION OF THE CONTEXT LEMMA

% Definition of natural numbers
LF nat : type =
  | z : nat
  | succ : nat -> nat
;

LF eqnat : nat -> nat -> type =
  | refnat: eqnat n n
;

LF leq : nat -> nat -> type =
  | samenat: leq n n
  | less: leq m n -> leq m (succ n)
;

% size of a context
inductive Dim : {g:ctx} [|- nat] -> ctype =
  | Dim_z: Dim [] [|- z]
  | Dim_succ: Dim [g] [|- n] -> Dim [g,x:names] [|- succ n]
;

% Multi_step_tau [|- n] [g |- P] [g |- Q] iff P evolves to Q in exactly n tau-transitions
inductive Multi_step_tau : (g:ctx) [|- nat] -> [g |- proc] -> [g |- proc] -> ctype =
  | Ms_tau_z: Multi_step_tau [|- z] [g |- P] [g |- P]
  | Ms_tau_succ: [g |- fstep P1 f_tau P2] -> Multi_step_tau [|- n] [g |- P2] [g |- P3] -> Multi_step_tau [|- succ n] [g |- P1] [g |- P3]
;

% prefix of a context
inductive Ctx_Pref: {g1:ctx} {g3:ctx} ctype =
  | Ctx_Pref_same: Ctx_Pref [g1] [g1]
  | Ctx_Pref_add: Ctx_Pref [g1] [g3] -> Ctx_Pref [g1] [g3,x:names]
;

% concatenation of contexts
inductive Ctx_App: {g1:ctx} {g2:ctx} {g3:ctx} ctype =
  | Ctx_App_empty: Ctx_App [g1] [] [g1]
  | Ctx_App_add: Ctx_App [g1] [g2] [g3] -> Ctx_App [g1] [g2,x:names] [g3,x:names]
;

% same name in a bigger context
inductive Name_Exp: (g2:ctx) (g3:ctx) [g2 |- names] -> [g3 |- names] -> ctype =
  | Name_Exp_same: Name_Exp [g2 |- X] [g2 |- X]
  | Name_Exp_add: Name_Exp [g2 |- X2] [g3 |- X3] -> Name_Exp [g2 |- X2] [g3,x:names |- X3[..]]
;

% concatenation of substitutions
inductive Subst_App: (g1:ctx) (g2:ctx) (g3:ctx) $[g1 |- g1] -> $[g2 |- g2] -> $[g3 |- g3] -> ctype =
  | Subst_App_empty: Subst_App $[g1 |- $S1] $[|- ^] $[g1 |- $S1]
  | Subst_App_add: Subst_App $[g1 |- $S1] $[g2 |- $S2] $[g3 |- $S3] 
    -> Name_Exp [g2,x:names |- X2] [g3,x:names |- X3] 
    -> Subst_App $[g1 |- $S1] $[g2,x:names |- $S2[..],X2] $[g3,x:names |- $S3[..],X3]
;

% expansion of a substitution
inductive Subst_Exp: (g1:ctx) (g3:ctx) $[g1 |- g1] -> $[g3 |- g3] -> ctype =
  | Subst_Exp_Def: Subst_App $[g1 |- $S1] $[g2 |- ..] $[g3 |- $S3] -> Subst_Exp $[g1 |- $S1] $[g3 |- $S3]
;

% Context built from substitution
%{inductive Ctx_pref_zero: (g':ctx) {g:ctx} $[g' |- g'] -> [g',x:names |- proc] -> ctype =
  | Ctx_pref_zero_empty: Ctx_pref_zero [] $[g' |- ..] [g',x:names |- p_zero]
  | Ctx_pref_zero_add: Partial_subst [g] [g'] Ctx_pref_zero [g] $[g' |- $S] [g',x:names |- P0] -> Ctx_pref_zero [g,x:names] $[g |- $S,Y] [g,x:names |- p_out x Y[..] P0]
;

inductive Ctx_pref_proc: (g:ctx) {g':ctx} [g |- proc] -> [g,x:names |- proc] -> ctype =
  | Ctx_pref_proc_empty: Ctx_pref_proc [] [g |- P] [g,x:names |- P[..]]
  | Ctx_pref_proc_add: Ctx_pref_proc [g'] [g |- P] [g,x:names |- P'] -> Ctx_pref_proc [g',x:names] [g |- P] [g,x:names |- p_in x \x. P'[..]]
;}%

%{inductive Ctx_par_subst: (g':ctx) {g:ctx} $[g |- g] -> [g' |- proc] -> [g' |- proc] -> [g',x:names |- proc] -> ctype =
  | Ctx_par_subst_empty: Ctx_par_subst [] $[ |- ^] [g' |- R] [g' |- P] [g',x:names |- (p_zero p_par P[..]) p_par R[..]]
  | Ctx_par_subst_add: Ctx_par_subst [g] $[g |- $S] [g',x:names |- R] [g',x:names |- P] [g',x:names,y:names |- (Po p_par Pi) p_par R[..]] -> Name_Exp [g,x:names |- X] [g',x:names |- X'] -> Ctx_par_subst [g,x:names] $[g,x:names |- $S[..],X] [g',x:names |- R] [g',x:names |- P] [g',x:names,y:names |- ((p_out x X'[..] Po) p_par (p_in y \x. Pi))]
;}%
\end{lstlisting}
\vspace{-0.5\baselineskip}
\caption{Excerpt from encoding of contexts.}
\label{fig:PCtx}
\end{figure}

Structural congruence can alternatively be defined as a contextual
equivalence via the inductive type {\ttfamily \small
  \color{belugagreen}Cong} in Fig.~\ref{fig:strcong}, that uses 
 the type {\ttfamily \small
  \color{belugagreen}PCtx} in the \emph{context closure} clause.~\clink{1\_definitions.bel}{76}{94} 


\begin{figure}[th]
\begin{lstlisting}[style=belugastyleframes,linerange={Structural_congruence-End},basicstyle=\footnotesize\ttfamily]
%%%%%%%%%%%%%%%%%%%%%%%%%%%%%%%%%%%%%%%%%%%%%%%%%%%%%%%%%

%%% Definitions %%%
% In order to encode input and restriction processes, which both bind names, we exploit higher order abstract syntax
% and use higher-order functions from names to proc.
% In this way we don't need to give an explicit name to bound names and deal with alpha-renaming or substitution.

%% Names_and_Processes %%
LF names: type =;

LF proc: type =
 | p_zero: proc                                
 | p_in: names → (names → proc) → proc         
 | p_out: names → names → proc → proc          
 | p_par: proc → proc → proc                   
 | p_res: (names → proc) → proc                
 | p_rep: proc → proc; ---infix p_par 11 left.

schema ctx = names;
%% End %%

% Equality of processes
LF eqp: proc → proc → type =
  | refp: eqp X X
;

% Equality of names
LF eqn: names → names → type =
  | refl: eqn N N
;

% empty type
LF not_possible : type = 
;



% Beniamino's trick - Predicate for contexts for a couple of terms
% Meta level
% PCtx [g' |- P] [g |- CP] [g' |- Q] [g |- CQ] holds if there is a context C such that C[P] = CP and C[Q] = CQ

%% Contexts %%
inductive PCtx: (g:ctx)(g':ctx)[g' ⊢ proc] → [g ⊢ proc] → [g' ⊢ proc] → [g ⊢ proc] → ctype =
  | pidM: PCtx [g ⊢ P] [g ⊢ P] [g ⊢ Q] [g ⊢ Q]
  | pinM: PCtx [g' ⊢ P] [g,x:names ⊢ CP] [g' ⊢ Q] [g,x:names ⊢ CQ] 
       → PCtx [g' ⊢ P] [g ⊢ p_in X (\x.CP)] [g' ⊢ Q] [g ⊢ p_in X (\x.CQ)]
  | poutM: PCtx [g' ⊢ P] [g ⊢ CP] [g' ⊢ Q] [g ⊢ CQ] 
        → PCtx [g' ⊢ P] [g ⊢ p_out X Y CP] [g' ⊢ Q] [g ⊢ p_out X Y CQ]
  ...								 
%% End %%

% Structural congruence
LF cong: proc → proc → type =
% Abelian Monoid Laws for Parallel Composition
  | par_assoc: cong (P p_par (Q p_par R)) ((P p_par Q) p_par R)
  | par_unit: cong (P p_par p_zero) P
  | par_comm: cong (P p_par Q) (Q p_par P)
% Scope Extension Laws
  | sc_ext_zero: cong (p_res \x.p_zero) p_zero
  | sc_ext_par: cong ((p_res P) p_par Q) (p_res \x.((P x) p_par Q))
  | sc_ext_res: cong (p_res \x.(p_res \y.(P x y))) (p_res \y.(p_res \x.(P x y)))
% Replication unfolding
  | rep_unfold: cong (p_rep P) (P p_par (p_rep P))
% Compatibility Laws
  | c_in: ({y:names} cong (P y) (Q y)) → cong (p_in X P) (p_in X Q)
  | c_out: cong P Q → cong (p_out X Y P) (p_out X Y Q)
  | c_par: cong P P' → cong (P p_par Q) (P' p_par Q)
  | c_res: ({x:names} cong (P x) (Q x)) → cong (p_res P) (p_res Q)
  | c_rep: cong P Q -> cong (p_rep P) (p_rep Q)
% Equivalence Relation Laws
  | c_ref: cong P P
  | c_sym: cong P Q → cong Q P
  | c_trans: cong P Q → cong Q R → cong P R
;
--infix cong 11 left.

%% Structural_congruence %%
inductive Cong: (g:ctx) [g ⊢ proc] → [g ⊢ proc] → ctype =
  ...
% Replication unfolding
  | Rep_unfold: Cong [g ⊢ (p_rep P)] [g ⊢ (P p_par (p_rep P))]
% Context closure
  | C_ctx: Cong [g' ⊢ P] [g' ⊢ Q] → PCtx [g' ⊢ P] [g ⊢ CP] [g' ⊢ Q] [g ⊢ CQ] 
        → Cong [g ⊢ CP] [g ⊢ CQ]
% Equivalence Relation Laws
  | C_ref: Cong [g ⊢ P] [g ⊢ P]
  | C_sym: Cong [g ⊢ P] [g ⊢ Q] → Cong [g ⊢ Q] [g ⊢ P]
  | C_trans: Cong [g ⊢ P] [g ⊢ Q] → Cong [g ⊢ Q] [g ⊢ R] → Cong [g ⊢ P] [g ⊢ R]
;
%% End %%

%%%%%%%%%%%%%%%%%%%%%%%%%%%%%%%%%%%%%%%%%%%%%%%%%%%%%%%%%

%%% LTS Semantics %%%
% We follow "Pi-Calculus in (Co)Inductive Type Theory" [Honsell et al. 2001] for the encoding of late LTS semantics.
% We define two different types for free/bound actions, and two different relations for transitions via free/bound actions.
% The result of a free transition is a process, while the result of a bound transition is a function from names to processes:
% instead of stating the bound name involved in the transition explicitly, that name is the argument of the aforementioned function.

% Free Actions
LF f_act: type =
  | f_tau: f_act
  | f_out: names → names → f_act
;

% Bound Actions
LF b_act: type =
  | b_in: names → b_act
  | b_out: names → b_act
;

% Equality of actions
LF eqf: f_act → f_act → type =
  | reff: eqf A A
;

LF eqb: b_act → b_act → type =
  | refb: eqb A A
;

% Transition Relation
LF fstep: proc → f_act → proc → type =
  | fs_out: fstep (p_out X Y P) (f_out X Y) P
  | fs_par1: fstep P A P' → fstep (P p_par Q) A (P' p_par Q)
  | fs_par2: fstep Q A Q' → fstep (P p_par Q) A (P p_par Q')
  | fs_com1: fstep P (f_out X Y) P' → bstep Q (b_in X) Q'
    → fstep (P p_par Q) f_tau (P' p_par (Q' Y))
  | fs_com2: bstep P (b_in X) P' → fstep Q (f_out X Y) Q'
    → fstep (P p_par Q) f_tau ((P' Y) p_par Q')
  | fs_res: ({z:names} fstep (P z) A (P' z))
    → fstep (p_res P) A (p_res P')
  | fs_close1: bstep P (b_out X) P' → bstep Q (b_in X) Q'
    → fstep (P p_par Q) f_tau (p_res \z.((P' z) p_par (Q' z)))
  | fs_close2: bstep P (b_in X) P' → bstep Q (b_out X) Q'
    → fstep (P p_par Q) f_tau (p_res \z.((P' z) p_par (Q' z)))
  | fs_rep: fstep P A P' → fstep (p_rep P) A (P' p_par (p_rep P))
  | fs_rep_com: fstep P (f_out X Y) P1 → bstep P (b_in X) P2
    → fstep (p_rep P) f_tau ((P1 p_par (P2 Y)) p_par (p_rep P))
  | fs_rep_close: bstep P (b_out X) P1 → bstep P (b_in X) P2
    → fstep (p_rep P) f_tau ((p_res \z.((P1 z) p_par (P2 z))) p_par (p_rep P))

and bstep: proc → b_act → (names → proc) → type =
  | bs_in: bstep (p_in X P) (b_in X) P
  | bs_par1: bstep P A P' → bstep (P p_par Q) A \x.((P' x) p_par Q)
  | bs_par2: bstep Q A Q' → bstep (P p_par Q) A \x.(P p_par (Q' x))
  | bs_res: ({z:names} bstep (P z) A (P' z))
    → bstep (p_res P) A \x.(p_res \z.(P' z x))
  | bs_open: ({z:names} fstep (P z) (f_out X z) (P' z))
    → bstep (p_res P) (b_out X) P'
  | bs_rep: bstep P A P' → bstep (p_rep P) A \x.((P' x) p_par (p_rep P))
;


%%% Barbed similarity %%%

inductive Barb_in : (g:ctx) [g |- proc] -> [g |- names] -> ctype =
  | Barb_in_b : [g |- bstep P (b_in X) (\x. P')] -> Barb_in [g |- P] [g |- X]
;

inductive Barb_out : (g:ctx) [g |- proc] -> [g |- names] -> ctype =
  | Barb_out_f : [g |- fstep P (f_out X Y) P'] -> Barb_out [g |- P] [g |- X]
  | Barb_out_b : [g |- bstep P (b_out X) (\x. P')] -> Barb_out [g |- P] [g |- X]
;

LF barb_in_rew : proc -> names -> type =
  | biw_zero : barb_in_rew ((p_in X Q) p_par R) X
  | biw_add: ({x : names} barb_in_rew (P x) X) -> barb_in_rew (p_res P) X
;

LF barb_in_red : proc -> names -> type =
  | barb_in_red_def: P cong P' -> barb_in_rew P' X -> barb_in_red P X
;

LF barb_out_rew : proc -> names -> type =
  | bow_zero : barb_out_rew ((p_out X Y Q) p_par R) X
  | bow_add: ({x : names} barb_out_rew (P x) X) -> barb_out_rew (p_res P) X
;

LF barb_out_red : proc -> names -> type =
  | barb_out_red_def: P cong P' -> barb_out_rew P' X -> barb_out_red P X
;

%% BarbSim %%
coinductive BarbSim : (g:ctx) [g ⊢ proc] → [g ⊢ proc] → ctype =
  | (BarbSim_barb_in : BarbSim [g ⊢ P] [g ⊢ Q])
              :: Barb_in [g ⊢ P] [g ⊢ X] → Barb_in [g ⊢ Q] [g ⊢ X]
  | (BarbSim_barb_out : BarbSim [g ⊢ P] [g ⊢ Q])
              :: Barb_out [g ⊢ P] [g ⊢ X] → Barb_out [g ⊢ Q] [g ⊢ X]
  | (BarbSim_tau : BarbSim [g ⊢ P] [g ⊢ Q])
              :: [g ⊢ fstep P f_tau P'] → Ex_sim_barb [g ⊢ P] [g ⊢ Q] [g ⊢ P']

and inductive Ex_sim_barb : (g:ctx) [g ⊢ proc] → [g ⊢ proc] → [g ⊢ proc] → ctype =
  | Ex_sim_barb_def : [g ⊢ fstep Q f_tau Q'] → BarbSim [g ⊢ P'] [g ⊢ Q'] 
                   → Ex_sim_barb [g ⊢ P] [g ⊢ Q] [g ⊢ P']
;
%% End %%

%%% Barbed similarity up to%%%

coinductive BarbSimUpTo : (g:ctx) [g |- proc] -> [g |- proc] -> ctype =
  | (BarbSimUpTo_barb_in : BarbSimUpTo [g |- P] [g |- Q])
              :: Barb_in [g |- P] [g |- X] -> Barb_in [g |- Q] [g |- X]
  | (BarbSimUpTo_barb_out : BarbSimUpTo [g |- P] [g |- Q])
              :: Barb_out [g |- P] [g |- X] -> Barb_out [g |- Q] [g |- X]
  | (BarbSimUpTo_f_tau : BarbSimUpTo [g |- S1] [g |- S2])
              :: [g |- fstep S1 f_tau S1'] -> Ex_sim_barb_up_to [g |- S1] [g |- S2] [g |- S1']

and inductive Ex_sim_barb_up_to : (g:ctx) [g |- proc] -> [g |- proc] -> [g |- proc] -> ctype =
  | Ex_sim_barb_up_to_def : [g |- fstep Q f_tau Q'] -> BarbSim [g |- P'] [g |- P''] -> BarbSimUpTo [g |- P''] [g |- Q''] -> BarbSim [g |- Q''] [g |- Q'] 
                 -> Ex_sim_barb_up_to [g |- P] [g |- Q] [g |- P']
;

%%% Barbed precongruence %%%
inductive BarbPre: (g':ctx) [g' |- proc] -> [g' |- proc] -> ctype =
  | BC:  (g':ctx) {P:[g' |- proc]} {Q:[g' |- proc]} 
    ((g : ctx) {CP:[g |- proc]} {CQ:[g |- proc]} PCtx [g' ⊢ P] [g ⊢ CP] [g' ⊢ Q] [g ⊢ CQ] -> BarbSim [g |- CP] [g |- CQ]) 
    -> BarbPre [g' |- P] [g' |- Q]
;

%%% Characterizations of barbed precongruence

% Context lemma characterization of barbed precongruence
inductive BarbPre': (g:ctx) [g |- proc] -> [g |- proc] -> ctype =
  | BC' : ((h:ctx) {$S : $[h |- g]} {R : [h |- proc]} -> BarbSim [h |- P[$S] p_par R] [h |- Q[$S] p_par R]) -> BarbPre' [g |- P] [g |- Q]
;

% Largest precongruence included in barbed similarity
coinductive LargPreCong: (g':ctx) [g' |- proc] -> [g' |- proc] -> ctype =
  | (LP_BarbSim : LargPreCong [g' |- P] [g' |- Q])
              :: BarbSim [g' |- P] [g' |- Q]
  | (LP_PreCong : LargPreCong [g' |- P] [g' |- Q])
              :: PCtx [g' ⊢ P] [g ⊢ CP] [g' ⊢ Q] [g ⊢ CQ] -> LargPreCong [g ⊢ CP] [g ⊢ CQ]
;


%%% OTHER IMPLICATION OF THE CONTEXT LEMMA

% Definition of natural numbers
LF nat : type =
  | z : nat
  | succ : nat -> nat
;

LF eqnat : nat -> nat -> type =
  | refnat: eqnat n n
;

LF leq : nat -> nat -> type =
  | samenat: leq n n
  | less: leq m n -> leq m (succ n)
;

% size of a context
inductive Dim : {g:ctx} [|- nat] -> ctype =
  | Dim_z: Dim [] [|- z]
  | Dim_succ: Dim [g] [|- n] -> Dim [g,x:names] [|- succ n]
;

% Multi_step_tau [|- n] [g |- P] [g |- Q] iff P evolves to Q in exactly n tau-transitions
inductive Multi_step_tau : (g:ctx) [|- nat] -> [g |- proc] -> [g |- proc] -> ctype =
  | Ms_tau_z: Multi_step_tau [|- z] [g |- P] [g |- P]
  | Ms_tau_succ: [g |- fstep P1 f_tau P2] -> Multi_step_tau [|- n] [g |- P2] [g |- P3] -> Multi_step_tau [|- succ n] [g |- P1] [g |- P3]
;

% prefix of a context
inductive Ctx_Pref: {g1:ctx} {g3:ctx} ctype =
  | Ctx_Pref_same: Ctx_Pref [g1] [g1]
  | Ctx_Pref_add: Ctx_Pref [g1] [g3] -> Ctx_Pref [g1] [g3,x:names]
;

% concatenation of contexts
inductive Ctx_App: {g1:ctx} {g2:ctx} {g3:ctx} ctype =
  | Ctx_App_empty: Ctx_App [g1] [] [g1]
  | Ctx_App_add: Ctx_App [g1] [g2] [g3] -> Ctx_App [g1] [g2,x:names] [g3,x:names]
;

% same name in a bigger context
inductive Name_Exp: (g2:ctx) (g3:ctx) [g2 |- names] -> [g3 |- names] -> ctype =
  | Name_Exp_same: Name_Exp [g2 |- X] [g2 |- X]
  | Name_Exp_add: Name_Exp [g2 |- X2] [g3 |- X3] -> Name_Exp [g2 |- X2] [g3,x:names |- X3[..]]
;

% concatenation of substitutions
inductive Subst_App: (g1:ctx) (g2:ctx) (g3:ctx) $[g1 |- g1] -> $[g2 |- g2] -> $[g3 |- g3] -> ctype =
  | Subst_App_empty: Subst_App $[g1 |- $S1] $[|- ^] $[g1 |- $S1]
  | Subst_App_add: Subst_App $[g1 |- $S1] $[g2 |- $S2] $[g3 |- $S3] 
    -> Name_Exp [g2,x:names |- X2] [g3,x:names |- X3] 
    -> Subst_App $[g1 |- $S1] $[g2,x:names |- $S2[..],X2] $[g3,x:names |- $S3[..],X3]
;

% expansion of a substitution
inductive Subst_Exp: (g1:ctx) (g3:ctx) $[g1 |- g1] -> $[g3 |- g3] -> ctype =
  | Subst_Exp_Def: Subst_App $[g1 |- $S1] $[g2 |- ..] $[g3 |- $S3] -> Subst_Exp $[g1 |- $S1] $[g3 |- $S3]
;

% Context built from substitution
%{inductive Ctx_pref_zero: (g':ctx) {g:ctx} $[g' |- g'] -> [g',x:names |- proc] -> ctype =
  | Ctx_pref_zero_empty: Ctx_pref_zero [] $[g' |- ..] [g',x:names |- p_zero]
  | Ctx_pref_zero_add: Partial_subst [g] [g'] Ctx_pref_zero [g] $[g' |- $S] [g',x:names |- P0] -> Ctx_pref_zero [g,x:names] $[g |- $S,Y] [g,x:names |- p_out x Y[..] P0]
;

inductive Ctx_pref_proc: (g:ctx) {g':ctx} [g |- proc] -> [g,x:names |- proc] -> ctype =
  | Ctx_pref_proc_empty: Ctx_pref_proc [] [g |- P] [g,x:names |- P[..]]
  | Ctx_pref_proc_add: Ctx_pref_proc [g'] [g |- P] [g,x:names |- P'] -> Ctx_pref_proc [g',x:names] [g |- P] [g,x:names |- p_in x \x. P'[..]]
;}%

%{inductive Ctx_par_subst: (g':ctx) {g:ctx} $[g |- g] -> [g' |- proc] -> [g' |- proc] -> [g',x:names |- proc] -> ctype =
  | Ctx_par_subst_empty: Ctx_par_subst [] $[ |- ^] [g' |- R] [g' |- P] [g',x:names |- (p_zero p_par P[..]) p_par R[..]]
  | Ctx_par_subst_add: Ctx_par_subst [g] $[g |- $S] [g',x:names |- R] [g',x:names |- P] [g',x:names,y:names |- (Po p_par Pi) p_par R[..]] -> Name_Exp [g,x:names |- X] [g',x:names |- X'] -> Ctx_par_subst [g,x:names] $[g,x:names |- $S[..],X] [g',x:names |- R] [g',x:names |- P] [g',x:names,y:names |- ((p_out x X'[..] Po) p_par (p_in y \x. Pi))]
;}%
\end{lstlisting}
\vspace{-0.5\baselineskip}
\caption{Excerpt from the inductive definition of structural congruence.}
\label{fig:strcong}
\end{figure}

\subsection{Preliminary Lemmas}
We start by establishing some basic properties of the LTS and of
process contexts. The former ones can be found
in~\cite{SWBook}, modulo our choice of late
semantics, while the latter ones, adapted from~\cite{LancelotAV25},
stem from the implementation of process contexts.

\begin{lemma}[Properties of the LTS]\label{lts}\leavevmode
  \begin{enumerate}[label=\arabic*., ref=\thelemma.\arabic*]
  \item\label{lts:one} If $P \xrightarrow{\overline{x}y} P'$ and $z$ is a name, then
either $z=y$ and
    $\nu z\, P \xrightarrow{\overline{x}(z)} P'$ or
    $\nu z\, P \xrightarrow{\overline{x}y} \nu z\, P'$.~\clink{2\_lts.bel}{3}{17}
    \item\label{lts:two} If $S \mid !P \xrightarrow{\tau} R$, then there is $Q$ such that $S \mid (P \mid P) \xrightarrow{ \tau} Q$ and $R \equiv Q \mid !P$.~\clink{2\_lts.bel}{20}{86}
    \item\label{lts:three} If $x\not\in\mathrm{fn}(P)$ and $P \xrightarrow{\alpha} P'$,
      then $x\not\in(\mathrm{fn}(\alpha)\cup\mathtt{fn}(P'))$.~\clink{6\_structcong\_in\_barbsim.bel}{3}{190}
    \item\label{strength} If $P \equiv Q$, then
    \begin{itemize}
      \item if $P \xrightarrow{\alpha} P'$, then there is $Q'$ such that $Q \xrightarrow{\alpha} Q'$ and $P' \equiv Q'$;
      \item if $Q \xrightarrow{\alpha} Q'$, then there is $P'$ such that $P \xrightarrow{\alpha} P'$ and $P' \equiv Q'$.~\clink{6\_structcong\_in\_barbsim.bel}{193}{720}
    \end{itemize}
  \end{enumerate}
\end{lemma}

\begin{proof} (Sketch)
  The first two statements are proven by a straightforward induction on the derivation of the transition in the hypothesis. The third  one is split into two lemmas, separating free and bound actions, which are proven by a mutual induction on the derivation of the transition in the hypothesis. The last one is split into four lemmas (it requires to prove two different theses and in both we need to distinguish free and bound actions), which are proven by a mutual induction on the derivation of the congruence in the hypothesis. We note
  that the third and the fourth results are extensions of those 
  in~\cite{Cecilia24}. 
\end{proof}
\smallskip

Now for properties of process contexts:  some of
them would be obvious with the on-paper definition of contexts, but are
not trivial in the {\ttfamily \small \color{belugagreen}PCtx}
formalization.  We recall that $((P,C_P),(Q,C_Q))$ means that there
is a context $C$ such that $C[P] = C_P$ and $C[Q] = C_Q$.

\begin{lemma}[Properties of contexts]\leavevmode
  \begin{enumerate}
    \item (Symmetry) If $((P,C_P),(Q,C_Q))$, then $((Q,C_Q),(P,C_P))$.~\clink{4\_contexts.bel}{4}{15}
    \item (Transitivity) If $((P,C_P),(R,C_R))$, then for all $Q$ there is $C_Q$ such that $((P,C_P),(Q,C_Q))$ and \allowbreak\ \mbox{$((Q,C_Q),(R,C_R))$}.~\clink{4\_contexts.bel}{19}{38} 
    \item (Functionality) If $((P,C_P),(P,C_P'))$, then $C_P = C_P'$.~\clink{4\_contexts.bel}{41}{50}
    \item (Context composition) If $((P,C_P),(Q,C_Q))$ and $((C_P,C_{C_P}),(C_Q,C_{C_Q}))$, then $((P,C_{C_P}), (Q,C_{C_Q}))$.~\clink{4\_contexts.bel}{54}{64}
  \end{enumerate}
\end{lemma}

\begin{proof}
  All proofs are by induction on the derivation of the context relation in the hypothesis.
\end{proof}

\smallskip
As is well known~\cite{Lassen98}, a binary relation on processes constitutes a contextual equivalence if and only if it is a compatible equivalence. Because Beluga’s logic lacks support for higher-order predicates, this meta-theorem cannot be generalized; rather, the equivalence must be formalized independently for each specific relation.

\begin{lemma}
  Structural congruence can be characterized either via compatibility~\clink{3\_cong\_equiv.bel}{4}{22}
  or by context
  closure.~\clink{3\_cong\_equiv.bel}{38}{53}
\end{lemma}
\begin{proof}
  (Sketch) Both inclusions are proved by induction on the derivation of the congruence rule in the hypothesis. The inclusion of the compatible relation in the contextual one first requires to  prove that the compatible relation is contextually closed too.~\clink{3\_cong\_equiv.bel}{25}{35}
\end{proof}

The LF definition is more convenient for discussing the properties of the LTS, as in~\cite{Cecilia24}, while the inductive one is used later to show the inclusion of structural congruence in another precongruence defined via contexts.


\section{Behavioral Equivalences}
\label{sec:beh}

Intuitively, two processes can be considered equivalent when they exhibit the same behavior. However, the level of detail at which behaviors can be distinguished depends on the chosen observational granularity; this is where a variety of bisimilarity notions come into play. In CCFB3, the authors focus on \emph{strong barbed bisimilarity}, 
which relates processes that can match internal ($\tau$) transitions and preserve the same \emph{barbs}, i.e.\ the ability to perform input or output on a given channel.


More precisely, the barb, or observability predicate, $P \downarrow_x$ (resp.\ $P \downarrow_{\overline{x}}$) 
holds if a process $P$ can perform an input action with subject $x$ (resp.\ an output action with subject $\overline{x}$). Following~\cite{SWBook},
this notion can be formalized in two equivalent ways, based either on
the structure of the process $P$ or on the (late) LTS:

\begin{enumerate}[label=\roman*)]
    \item\label{def:barb_one} $P \downarrow_x$ iff $P \equiv \nu z_1 \dots \nu z_n (x(y).Q \mid R)$ for some $y,z_1,\dots,z_n,Q,R$, with $x \notin \{z_1,\dots,z_n\}$. Analogously, $P \downarrow_{\overline{x}}$ iff $P \equiv \nu z_1 \dots \nu z_n (\overline{x}y.Q \mid R)$ for some $y,z_1,\dots,z_n,Q,R$, with $x \notin \{z_1,\dots,z_n\}$.
    
    The telescopes, i.e.\ $n$-ary sequences of binders, appearing in this definition of barb are encoded in Beluga by two types {\ttfamily \small \color{belugagreen}barb\_in\_rew} and {\ttfamily \small \color{belugagreen}barb\_out\_rew}, that build the sequence of restrictions incrementally.~\clink{1\_definitions.bel}{171}{179} Then, the two types encoding barbs identify processes congruent to such telescopes.~\clink{1\_definitions.bel}{181}{187}
    
    \item\label{def:barb_two} $P \downarrow_x$ iff $P \xrightarrow{x(z)} P'$ for some $z,P'$, and similarly $P \downarrow_{\overline{x}}$ iff $P \xrightarrow{\overline{x}y} P'$ or $P \xrightarrow{\overline{x}(z)} P'$ for some $y,z,P'$.
    
    We formalize  this at the meta level, but it could have been formulated as an LF type as easily. Since the LTS presents one input label and two output labels, {\ttfamily \small \color{belugagreen}Barb\_in} has one constructor, while {\ttfamily \small \color{belugagreen}Barb\_out} has two.~\clink{1\_definitions.bel}{161}{168}
\end{enumerate}



Definitions~\ref{def:barb_one} and~\ref{def:barb_two} are equivalent, in the sense given by the following lemma. For brevity, we omit some auxiliary technical lemmas and refer to the repository for full details.~\clink{7\_barbs\_equiv.bel}{6}{14}~\clink{7\_barbs\_equiv.bel}{24}{37}~\clink{7\_barbs\_equiv.bel}{71}{84}

\begin{lemma}\label{le:barb_equiv}
  Let $P$ be a process, $x$ be a name.
  The following are equivalent:
  \begin{enumerate}
    \item $P \equiv \nu z_1 \dots \nu z_n (x(y).Q \mid R)$ for some $y,z_1,\dots,z_n,Q,R$, with $x \notin \{z_1,\dots,z_n\}$;
    \item $P \xrightarrow{x(z)} P'$ for some $z,P'$.~\clink{7\_barbs\_equiv.bel}{17}{21}~\clink{7\_barbs\_equiv.bel}{40}{55}
  \end{enumerate}

  Similarly, the following are equivalent:
  \begin{enumerate}
    \item $P \equiv \nu z_1 \dots \nu z_n (\overline{x}y.Q \mid R)$ for some $y,z_1,\dots,z_n,Q,R$, with $x \notin \{z_1,\dots,z_n\}$;
    \item $P \xrightarrow{\overline{x}y} P'$ or $P \xrightarrow{\overline{x}(z)} P'$ for some $y,z,P'$.~\clink{7\_barbs\_equiv.bel}{58}{68}~\clink{7\_barbs\_equiv.bel}{87}{120}
  \end{enumerate}
  
\end{lemma}

\begin{proof}
  (Sketch)
  The proof that~\ref{def:barb_one} implies~\ref{def:barb_two} is a straightforward induction on the rewriting definitions. The converse proceeds by structural induction on $P$, followed by an inversion on the observable transition that stems from $P$; some subcases are handled by the omitted auxiliary lemmas.
\end{proof}

Being invariant w.r.t.\ the underlying LTS,
definition~\ref{def:barb_one} can be adopted even in the setting of
reduction semantics; however, definition~\ref{def:barb_two} has the benefit of
being more amenable to mechanized proofs without the need of algebraic
rewrites. For this reason, we use the latter in the
rest of the development.

A significant property of barbs is that they are preserved by contexts:

\begin{lemma}\label{le:barbs_ctx}
  \leavevmode
  \begin{enumerate}
    \item If $P \downarrow_x$ implies $Q \downarrow_x$ and $((P,C_P),(Q,C_Q))$, then $C_P \downarrow_x$ implies $C_Q \downarrow_x$.~\clink{4\_contexts.bel}{67}{83}
    \item If  $P \downarrow_{\overline{x}}$ implies $Q \downarrow_{\overline{x}}$ and $((P,C_P),(Q,C_Q))$, then $C_P \downarrow_{\overline{x}}$ implies $C_Q \downarrow_{\overline{x}}$.~\clink{4\_contexts.bel}{85}{164}
  \end{enumerate}
\end{lemma}
\begin{proof}
  By induction on the derivation of the context relation.
\end{proof}





\smallskip
We can now give the definition of barbed simulation.
A binary relation on processes $\mathcal{R}$ is a \textit{strong barbed simulation} if the following conditions hold:
\begin{enumerate}
    \item $\mathcal{R}$ is \textit{barb preserving}: if $P \mathbin\mathcal{R} Q$ and $P \downarrow_x$, then $Q \downarrow_x$; likewise, if $P \mathbin\mathcal{R} Q$ and $P \downarrow_{\overline{x}}$, then $Q \downarrow_{\overline{x}}$.
    \item $\mathcal{R}$ is a \textit{reduction simulation}: if $P \mathbin\mathcal{R} Q$ and $P \xrightarrow{\tau} P'$, then there is $Q'$ such that $Q \xrightarrow{\tau} Q'$ and $P' \mathbin\mathcal{R} Q'$.
\end{enumerate}
The union of all barbed simulations is called \textit{strong barbed similarity} and will be denoted by $\bsim$.

As explained above, barbs are not based on a specific LTS, and, as
proven in~\cite{Cecilia24}, albeit {without the replication operator},
early and late semantics provide the same $\tau$ actions. Thus, our
definition of strong barbed similarity turns out to be  equivalent to the one based
on the early semantics.

(Bi)simulations are perhaps the best known instance of coinductively defined relations. In calculi where the transition relation is well-founded --- e.g., those where each transition produces a structurally smaller process, ensuring that every computation trace is finite --- (bi)similarity can equivalently be characterized inductively. This is not the case in our fragment of the $\pi$-calculus, where infinite behaviors arise due to replication: thus, a genuinely coinductive treatment becomes necessary.

To our rescue, Beluga  supports coinductive reasoning in a very flexible way~\cite{ThibodeauCP16}.
Coinductive types are defined dually to inductive types: instead of constructors, they are characterized by destructors, or \emph{observations}. More specifically, while an inductive type is defined by inference rules whose conclusion is an instance of the type, a coinductive type is defined by inference rules where an instance of the type is the premise. Beluga separates this premise from the rest of the rule by using \texttt{::}.


Reasoning on coinductive objects proceeds via \emph{copattern matching}. By the Curry-Howard correspondence, proofs by (co)induction are encoded as total (co)recursive functions. In the coinductive case, totality requires coverage of all observations, while productivity ensures that each corecursive call is guarded by an observation; in the (current) absence of a productivity checker in Beluga, this condition must be verified manually. Observations of a coinductive object \texttt{c} are accessed using the syntax \texttt{c.observation\_name}, mirroring field access in mainstream programming languages.

As an example of a coinductive definition, we can look at the formalization of strong
barbed similarity in Fig.~\ref{fig:barb_sim}, which presents three observations, corresponding to
the input barbs, the output barbs and the $\tau$ transitions. Recall
that Beluga does not provide constructs for existential quantification, conjunction, or disjunction, so these must be encoded as inductive types.

\begin{figure}[th]
\begin{lstlisting}[style=belugastyleframes,linerange={BarbSim-End}]
%%%%%%%%%%%%%%%%%%%%%%%%%%%%%%%%%%%%%%%%%%%%%%%%%%%%%%%%%

%%% Definitions %%%
% In order to encode input and restriction processes, which both bind names, we exploit higher order abstract syntax
% and use higher-order functions from names to proc.
% In this way we don't need to give an explicit name to bound names and deal with alpha-renaming or substitution.

%% Names_and_Processes %%
LF names: type =;

LF proc: type =
 | p_zero: proc                                
 | p_in: names → (names → proc) → proc         
 | p_out: names → names → proc → proc          
 | p_par: proc → proc → proc                   
 | p_res: (names → proc) → proc                
 | p_rep: proc → proc; ---infix p_par 11 left.

schema ctx = names;
%% End %%

% Equality of processes
LF eqp: proc → proc → type =
  | refp: eqp X X
;

% Equality of names
LF eqn: names → names → type =
  | refl: eqn N N
;

% empty type
LF not_possible : type = 
;



% Beniamino's trick - Predicate for contexts for a couple of terms
% Meta level
% PCtx [g' |- P] [g |- CP] [g' |- Q] [g |- CQ] holds if there is a context C such that C[P] = CP and C[Q] = CQ

%% Contexts %%
inductive PCtx: (g:ctx)(g':ctx)[g' ⊢ proc] → [g ⊢ proc] → [g' ⊢ proc] → [g ⊢ proc] → ctype =
  | pidM: PCtx [g ⊢ P] [g ⊢ P] [g ⊢ Q] [g ⊢ Q]
  | pinM: PCtx [g' ⊢ P] [g,x:names ⊢ CP] [g' ⊢ Q] [g,x:names ⊢ CQ] 
       → PCtx [g' ⊢ P] [g ⊢ p_in X (\x.CP)] [g' ⊢ Q] [g ⊢ p_in X (\x.CQ)]
  | poutM: PCtx [g' ⊢ P] [g ⊢ CP] [g' ⊢ Q] [g ⊢ CQ] 
        → PCtx [g' ⊢ P] [g ⊢ p_out X Y CP] [g' ⊢ Q] [g ⊢ p_out X Y CQ]
  ...								 
%% End %%

% Structural congruence
LF cong: proc → proc → type =
% Abelian Monoid Laws for Parallel Composition
  | par_assoc: cong (P p_par (Q p_par R)) ((P p_par Q) p_par R)
  | par_unit: cong (P p_par p_zero) P
  | par_comm: cong (P p_par Q) (Q p_par P)
% Scope Extension Laws
  | sc_ext_zero: cong (p_res \x.p_zero) p_zero
  | sc_ext_par: cong ((p_res P) p_par Q) (p_res \x.((P x) p_par Q))
  | sc_ext_res: cong (p_res \x.(p_res \y.(P x y))) (p_res \y.(p_res \x.(P x y)))
% Replication unfolding
  | rep_unfold: cong (p_rep P) (P p_par (p_rep P))
% Compatibility Laws
  | c_in: ({y:names} cong (P y) (Q y)) → cong (p_in X P) (p_in X Q)
  | c_out: cong P Q → cong (p_out X Y P) (p_out X Y Q)
  | c_par: cong P P' → cong (P p_par Q) (P' p_par Q)
  | c_res: ({x:names} cong (P x) (Q x)) → cong (p_res P) (p_res Q)
  | c_rep: cong P Q -> cong (p_rep P) (p_rep Q)
% Equivalence Relation Laws
  | c_ref: cong P P
  | c_sym: cong P Q → cong Q P
  | c_trans: cong P Q → cong Q R → cong P R
;
--infix cong 11 left.

%% Structural_congruence %%
inductive Cong: (g:ctx) [g ⊢ proc] → [g ⊢ proc] → ctype =
  ...
% Replication unfolding
  | Rep_unfold: Cong [g ⊢ (p_rep P)] [g ⊢ (P p_par (p_rep P))]
% Context closure
  | C_ctx: Cong [g' ⊢ P] [g' ⊢ Q] → PCtx [g' ⊢ P] [g ⊢ CP] [g' ⊢ Q] [g ⊢ CQ] 
        → Cong [g ⊢ CP] [g ⊢ CQ]
% Equivalence Relation Laws
  | C_ref: Cong [g ⊢ P] [g ⊢ P]
  | C_sym: Cong [g ⊢ P] [g ⊢ Q] → Cong [g ⊢ Q] [g ⊢ P]
  | C_trans: Cong [g ⊢ P] [g ⊢ Q] → Cong [g ⊢ Q] [g ⊢ R] → Cong [g ⊢ P] [g ⊢ R]
;
%% End %%

%%%%%%%%%%%%%%%%%%%%%%%%%%%%%%%%%%%%%%%%%%%%%%%%%%%%%%%%%

%%% LTS Semantics %%%
% We follow "Pi-Calculus in (Co)Inductive Type Theory" [Honsell et al. 2001] for the encoding of late LTS semantics.
% We define two different types for free/bound actions, and two different relations for transitions via free/bound actions.
% The result of a free transition is a process, while the result of a bound transition is a function from names to processes:
% instead of stating the bound name involved in the transition explicitly, that name is the argument of the aforementioned function.

% Free Actions
LF f_act: type =
  | f_tau: f_act
  | f_out: names → names → f_act
;

% Bound Actions
LF b_act: type =
  | b_in: names → b_act
  | b_out: names → b_act
;

% Equality of actions
LF eqf: f_act → f_act → type =
  | reff: eqf A A
;

LF eqb: b_act → b_act → type =
  | refb: eqb A A
;

% Transition Relation
LF fstep: proc → f_act → proc → type =
  | fs_out: fstep (p_out X Y P) (f_out X Y) P
  | fs_par1: fstep P A P' → fstep (P p_par Q) A (P' p_par Q)
  | fs_par2: fstep Q A Q' → fstep (P p_par Q) A (P p_par Q')
  | fs_com1: fstep P (f_out X Y) P' → bstep Q (b_in X) Q'
    → fstep (P p_par Q) f_tau (P' p_par (Q' Y))
  | fs_com2: bstep P (b_in X) P' → fstep Q (f_out X Y) Q'
    → fstep (P p_par Q) f_tau ((P' Y) p_par Q')
  | fs_res: ({z:names} fstep (P z) A (P' z))
    → fstep (p_res P) A (p_res P')
  | fs_close1: bstep P (b_out X) P' → bstep Q (b_in X) Q'
    → fstep (P p_par Q) f_tau (p_res \z.((P' z) p_par (Q' z)))
  | fs_close2: bstep P (b_in X) P' → bstep Q (b_out X) Q'
    → fstep (P p_par Q) f_tau (p_res \z.((P' z) p_par (Q' z)))
  | fs_rep: fstep P A P' → fstep (p_rep P) A (P' p_par (p_rep P))
  | fs_rep_com: fstep P (f_out X Y) P1 → bstep P (b_in X) P2
    → fstep (p_rep P) f_tau ((P1 p_par (P2 Y)) p_par (p_rep P))
  | fs_rep_close: bstep P (b_out X) P1 → bstep P (b_in X) P2
    → fstep (p_rep P) f_tau ((p_res \z.((P1 z) p_par (P2 z))) p_par (p_rep P))

and bstep: proc → b_act → (names → proc) → type =
  | bs_in: bstep (p_in X P) (b_in X) P
  | bs_par1: bstep P A P' → bstep (P p_par Q) A \x.((P' x) p_par Q)
  | bs_par2: bstep Q A Q' → bstep (P p_par Q) A \x.(P p_par (Q' x))
  | bs_res: ({z:names} bstep (P z) A (P' z))
    → bstep (p_res P) A \x.(p_res \z.(P' z x))
  | bs_open: ({z:names} fstep (P z) (f_out X z) (P' z))
    → bstep (p_res P) (b_out X) P'
  | bs_rep: bstep P A P' → bstep (p_rep P) A \x.((P' x) p_par (p_rep P))
;


%%% Barbed similarity %%%

inductive Barb_in : (g:ctx) [g |- proc] -> [g |- names] -> ctype =
  | Barb_in_b : [g |- bstep P (b_in X) (\x. P')] -> Barb_in [g |- P] [g |- X]
;

inductive Barb_out : (g:ctx) [g |- proc] -> [g |- names] -> ctype =
  | Barb_out_f : [g |- fstep P (f_out X Y) P'] -> Barb_out [g |- P] [g |- X]
  | Barb_out_b : [g |- bstep P (b_out X) (\x. P')] -> Barb_out [g |- P] [g |- X]
;

LF barb_in_rew : proc -> names -> type =
  | biw_zero : barb_in_rew ((p_in X Q) p_par R) X
  | biw_add: ({x : names} barb_in_rew (P x) X) -> barb_in_rew (p_res P) X
;

LF barb_in_red : proc -> names -> type =
  | barb_in_red_def: P cong P' -> barb_in_rew P' X -> barb_in_red P X
;

LF barb_out_rew : proc -> names -> type =
  | bow_zero : barb_out_rew ((p_out X Y Q) p_par R) X
  | bow_add: ({x : names} barb_out_rew (P x) X) -> barb_out_rew (p_res P) X
;

LF barb_out_red : proc -> names -> type =
  | barb_out_red_def: P cong P' -> barb_out_rew P' X -> barb_out_red P X
;

%% BarbSim %%
coinductive BarbSim : (g:ctx) [g ⊢ proc] → [g ⊢ proc] → ctype =
  | (BarbSim_barb_in : BarbSim [g ⊢ P] [g ⊢ Q])
              :: Barb_in [g ⊢ P] [g ⊢ X] → Barb_in [g ⊢ Q] [g ⊢ X]
  | (BarbSim_barb_out : BarbSim [g ⊢ P] [g ⊢ Q])
              :: Barb_out [g ⊢ P] [g ⊢ X] → Barb_out [g ⊢ Q] [g ⊢ X]
  | (BarbSim_tau : BarbSim [g ⊢ P] [g ⊢ Q])
              :: [g ⊢ fstep P f_tau P'] → Ex_sim_barb [g ⊢ P] [g ⊢ Q] [g ⊢ P']

and inductive Ex_sim_barb : (g:ctx) [g ⊢ proc] → [g ⊢ proc] → [g ⊢ proc] → ctype =
  | Ex_sim_barb_def : [g ⊢ fstep Q f_tau Q'] → BarbSim [g ⊢ P'] [g ⊢ Q'] 
                   → Ex_sim_barb [g ⊢ P] [g ⊢ Q] [g ⊢ P']
;
%% End %%

%%% Barbed similarity up to%%%

coinductive BarbSimUpTo : (g:ctx) [g |- proc] -> [g |- proc] -> ctype =
  | (BarbSimUpTo_barb_in : BarbSimUpTo [g |- P] [g |- Q])
              :: Barb_in [g |- P] [g |- X] -> Barb_in [g |- Q] [g |- X]
  | (BarbSimUpTo_barb_out : BarbSimUpTo [g |- P] [g |- Q])
              :: Barb_out [g |- P] [g |- X] -> Barb_out [g |- Q] [g |- X]
  | (BarbSimUpTo_f_tau : BarbSimUpTo [g |- S1] [g |- S2])
              :: [g |- fstep S1 f_tau S1'] -> Ex_sim_barb_up_to [g |- S1] [g |- S2] [g |- S1']

and inductive Ex_sim_barb_up_to : (g:ctx) [g |- proc] -> [g |- proc] -> [g |- proc] -> ctype =
  | Ex_sim_barb_up_to_def : [g |- fstep Q f_tau Q'] -> BarbSim [g |- P'] [g |- P''] -> BarbSimUpTo [g |- P''] [g |- Q''] -> BarbSim [g |- Q''] [g |- Q'] 
                 -> Ex_sim_barb_up_to [g |- P] [g |- Q] [g |- P']
;

%%% Barbed precongruence %%%
inductive BarbPre: (g':ctx) [g' |- proc] -> [g' |- proc] -> ctype =
  | BC:  (g':ctx) {P:[g' |- proc]} {Q:[g' |- proc]} 
    ((g : ctx) {CP:[g |- proc]} {CQ:[g |- proc]} PCtx [g' ⊢ P] [g ⊢ CP] [g' ⊢ Q] [g ⊢ CQ] -> BarbSim [g |- CP] [g |- CQ]) 
    -> BarbPre [g' |- P] [g' |- Q]
;

%%% Characterizations of barbed precongruence

% Context lemma characterization of barbed precongruence
inductive BarbPre': (g:ctx) [g |- proc] -> [g |- proc] -> ctype =
  | BC' : ((h:ctx) {$S : $[h |- g]} {R : [h |- proc]} -> BarbSim [h |- P[$S] p_par R] [h |- Q[$S] p_par R]) -> BarbPre' [g |- P] [g |- Q]
;

% Largest precongruence included in barbed similarity
coinductive LargPreCong: (g':ctx) [g' |- proc] -> [g' |- proc] -> ctype =
  | (LP_BarbSim : LargPreCong [g' |- P] [g' |- Q])
              :: BarbSim [g' |- P] [g' |- Q]
  | (LP_PreCong : LargPreCong [g' |- P] [g' |- Q])
              :: PCtx [g' ⊢ P] [g ⊢ CP] [g' ⊢ Q] [g ⊢ CQ] -> LargPreCong [g ⊢ CP] [g ⊢ CQ]
;


%%% OTHER IMPLICATION OF THE CONTEXT LEMMA

% Definition of natural numbers
LF nat : type =
  | z : nat
  | succ : nat -> nat
;

LF eqnat : nat -> nat -> type =
  | refnat: eqnat n n
;

LF leq : nat -> nat -> type =
  | samenat: leq n n
  | less: leq m n -> leq m (succ n)
;

% size of a context
inductive Dim : {g:ctx} [|- nat] -> ctype =
  | Dim_z: Dim [] [|- z]
  | Dim_succ: Dim [g] [|- n] -> Dim [g,x:names] [|- succ n]
;

% Multi_step_tau [|- n] [g |- P] [g |- Q] iff P evolves to Q in exactly n tau-transitions
inductive Multi_step_tau : (g:ctx) [|- nat] -> [g |- proc] -> [g |- proc] -> ctype =
  | Ms_tau_z: Multi_step_tau [|- z] [g |- P] [g |- P]
  | Ms_tau_succ: [g |- fstep P1 f_tau P2] -> Multi_step_tau [|- n] [g |- P2] [g |- P3] -> Multi_step_tau [|- succ n] [g |- P1] [g |- P3]
;

% prefix of a context
inductive Ctx_Pref: {g1:ctx} {g3:ctx} ctype =
  | Ctx_Pref_same: Ctx_Pref [g1] [g1]
  | Ctx_Pref_add: Ctx_Pref [g1] [g3] -> Ctx_Pref [g1] [g3,x:names]
;

% concatenation of contexts
inductive Ctx_App: {g1:ctx} {g2:ctx} {g3:ctx} ctype =
  | Ctx_App_empty: Ctx_App [g1] [] [g1]
  | Ctx_App_add: Ctx_App [g1] [g2] [g3] -> Ctx_App [g1] [g2,x:names] [g3,x:names]
;

% same name in a bigger context
inductive Name_Exp: (g2:ctx) (g3:ctx) [g2 |- names] -> [g3 |- names] -> ctype =
  | Name_Exp_same: Name_Exp [g2 |- X] [g2 |- X]
  | Name_Exp_add: Name_Exp [g2 |- X2] [g3 |- X3] -> Name_Exp [g2 |- X2] [g3,x:names |- X3[..]]
;

% concatenation of substitutions
inductive Subst_App: (g1:ctx) (g2:ctx) (g3:ctx) $[g1 |- g1] -> $[g2 |- g2] -> $[g3 |- g3] -> ctype =
  | Subst_App_empty: Subst_App $[g1 |- $S1] $[|- ^] $[g1 |- $S1]
  | Subst_App_add: Subst_App $[g1 |- $S1] $[g2 |- $S2] $[g3 |- $S3] 
    -> Name_Exp [g2,x:names |- X2] [g3,x:names |- X3] 
    -> Subst_App $[g1 |- $S1] $[g2,x:names |- $S2[..],X2] $[g3,x:names |- $S3[..],X3]
;

% expansion of a substitution
inductive Subst_Exp: (g1:ctx) (g3:ctx) $[g1 |- g1] -> $[g3 |- g3] -> ctype =
  | Subst_Exp_Def: Subst_App $[g1 |- $S1] $[g2 |- ..] $[g3 |- $S3] -> Subst_Exp $[g1 |- $S1] $[g3 |- $S3]
;

% Context built from substitution
%{inductive Ctx_pref_zero: (g':ctx) {g:ctx} $[g' |- g'] -> [g',x:names |- proc] -> ctype =
  | Ctx_pref_zero_empty: Ctx_pref_zero [] $[g' |- ..] [g',x:names |- p_zero]
  | Ctx_pref_zero_add: Partial_subst [g] [g'] Ctx_pref_zero [g] $[g' |- $S] [g',x:names |- P0] -> Ctx_pref_zero [g,x:names] $[g |- $S,Y] [g,x:names |- p_out x Y[..] P0]
;

inductive Ctx_pref_proc: (g:ctx) {g':ctx} [g |- proc] -> [g,x:names |- proc] -> ctype =
  | Ctx_pref_proc_empty: Ctx_pref_proc [] [g |- P] [g,x:names |- P[..]]
  | Ctx_pref_proc_add: Ctx_pref_proc [g'] [g |- P] [g,x:names |- P'] -> Ctx_pref_proc [g',x:names] [g |- P] [g,x:names |- p_in x \x. P'[..]]
;}%

%{inductive Ctx_par_subst: (g':ctx) {g:ctx} $[g |- g] -> [g' |- proc] -> [g' |- proc] -> [g',x:names |- proc] -> ctype =
  | Ctx_par_subst_empty: Ctx_par_subst [] $[ |- ^] [g' |- R] [g' |- P] [g',x:names |- (p_zero p_par P[..]) p_par R[..]]
  | Ctx_par_subst_add: Ctx_par_subst [g] $[g |- $S] [g',x:names |- R] [g',x:names |- P] [g',x:names,y:names |- (Po p_par Pi) p_par R[..]] -> Name_Exp [g,x:names |- X] [g',x:names |- X'] -> Ctx_par_subst [g,x:names] $[g,x:names |- $S[..],X] [g',x:names |- R] [g',x:names |- P] [g',x:names,y:names |- ((p_out x X'[..] Po) p_par (p_in y \x. Pi))]
;}%
\end{lstlisting}
\vspace{-0.5\baselineskip}
\caption{Coinductive definition of strong barbed similarity.}
\label{fig:barb_sim}
\end{figure}

\begin{lemma}[Properties of strong barbed similarity]
  \label{le:barbsim}
  Strong barbed similarity 
    \leavevmode
    \begin{enumerate}[label=\arabic*., ref=\thelemma.\arabic*]
        \item  is a preorder;~\clink{5\_barbsim.bel}{5}{15}
        \item\label{le:barbsim:two}  is preserved by input prefix, output prefix and restriction;~\clink{5\_barbsim.bel}{18}{40}
        \item  includes structural congruence.~\clink{6\_structcong\_in\_barbsim.bel}{724}{737} 
    \end{enumerate}
\end{lemma}
\begin{proof}
  \mbox{}
  
  \begin{enumerate}
  \item By a straightforward coinduction,  as in~\cite{MomiglianoPT19}.
  \item To illustrate how Beluga implements a
    coinductive proof, we detail the restriction case. On paper, we
    would prove that $\{(\nu z P, \nu z Q) \mid P \bsim Q\}$ is a
    strong barbed simulation and therefore is included in $\bsim$. In
    Beluga, this is realized by the recursive function in
    Fig.~\ref{fig:coind_ex}, whose signature encodes the statement. The proof distinguishes cases on the possible observations --- note the  copattern syntax reminiscent of record selection:
    \begin{itemize}
    \item In the barbs cases we employ Lemma~\ref{le:barbs_ctx}, by which if
      the barbs of $P$ are included in the barbs of $Q$, then for each
      context $C$ (and in particular $\nu z [\_]$), the barbs of
      $C[P]$ are included in the barbs of $C[Q]$.
    \item  More interestingly, in the transition case we invert on the
    restriction rule for a $\tau$ action from $P$ to some $P'$ (the first \lstinline{let}). Then,
    using the hypothesis $P \bsim Q$ (\lstinline{BarbSim}), we can find a matching action
    from $Q$ to $Q'$, where $P' \bsim Q'$. Finally, we can
    corecursively call the theorem on $P' \bsim Q'$ and, by applying
    the restriction rule forward, obtain the thesis. The corecursive
    call is valid, because it is guarded by the \lstinline{BarbSim_tau} observation
    and the proof covers all cases, because we analyzed all the
    observable properties that characterize strong barbed similarity.
    \end{itemize}


  \item This follows immediately from Lemma~\ref{strength}, which was also crucial for one implication of the Harmony Lemma. 

  \end{enumerate}

\end{proof}
\begin{figure}[ht]
  \lstinputlisting[style=belugastyleframes,linerange={Example-End}]{code/barbsim.bel}
\vspace{-0.5\baselineskip}
\caption{Encoding of the proof of Lemma~\ref{le:barbsim:two}.}
\label{fig:coind_ex}
\end{figure}

We now recall the notion of \emph{strong barbed simulation up to} $\bsim$, namely a 
 binary relation on processes $\mathcal{R}$ such that:
\begin{enumerate}
    \item $\mathcal{R}$ is barb preserving.
    \item if $P \mathbin\mathcal{R} Q$ and $P \xrightarrow{\tau} P'$, then there is $Q'$ such that $Q \xrightarrow{\tau} Q'$ and $P' \bsim \mathbin\mathcal{R} \bsim Q'$.
\end{enumerate}
As above, \emph{strong barbed similarity up to $\bsim$} is the union of all strong barbed
simulations up to $\bsim$.

Again, we cannot
develop a general theory of relations up to: we simply encode strong
barbed similarity up to $\bsim$ as we did for strong barbed
similarity --- the only difference being in the inductive auxiliary type,
which encodes the \emph{up to} $\bsim$.\clink{1\_definitions.bel}{204}{215} 

In this setup, strong barbed similarity up to $\bsim$ coincides with
$\bsim$.~\clink{5\_barbsim.bel}{44}{54} The use here of  \emph{up-to technique}
is just for convenience: in order to prove that a relation is included
in strong barbed similarity, it is sufficient to prove that it is a
strong barbed simulation up to $\bsim$ and therefore included in
barbed similarity up to $\bsim$.

\subsection{Barbed Precongruence}

One way to distinguish two processes is to observe whether they exhibit the same behavior when placed in the same environment, or context; a desirable property of a (bi)similarity is to be preserved by every possible environment. 
Unfortunately, strong barbed similarity does not satisfy contextual closure. 
For instance, the two processes $\overline{x}y.\overline{a}b.0$ and $\overline{x}y.0$ are strong barbed similar, since their only barb is $\downarrow_{\overline{x}}$ and they cannot perform any internal transition. Conversely, plugging them into the context $[\_] \mid x(z).0$ modifies their observable behavior: the former can perform a transition $\overline{x}y.\overline{a}b.0 \mid x(z).0 \xrightarrow{\tau} \overline{a}b.0 \mid 0$, enabling the observation of a barb $\downarrow_{\overline{a}}$ that the process $\overline{x}y.0 \mid x(z).0$, after an internal transition, cannot produce.

Therefore, we aim to identify those
 pairs of processes whose behavior remains strongly barbed similar under all contexts.
The processes $P$ and $Q$ are \textit{strongly barbed precongruent}, denoted by $P \bpre Q$, if, for each context $C$, $C[P] \bsim C[Q]$.~\clink{1\_definitions.bel}{218}{222}


Among the properties of strong barbed precongruence, the most prominent is its characterization as the largest precongruence included in strong barbed similarity. It turns out that the latter can be formalized as a coinductive definition with only two observations, requiring the inclusion in $\bsim$ and the contextual closure.~\clink{1\_definitions.bel}{230}{235}

%

\begin{lemma}[Properties of strong barbed precongruence] Strong barbed precongruence
    \leavevmode
    \begin{enumerate}[label=\arabic*., ref=\thelemma.\arabic*]
        \item is  a preorder;~\clink{8\_barbprecong.bel}{4}{13}
        \item is included in strong barbed similarity;~\clink{8\_barbprecong.bel}{16}{19}
        \item is a precongruence;~\clink{8\_barbprecong.bel}{32}{35}
        \item\label{lem:larg_precong} is the largest precongruence included in strong barbed similarity;~\clink{8\_barbprecong.bel}{37}{46}
        \item includes structural congruence.~\clink{8\_barbprecong.bel}{23}{27}
    \end{enumerate}
\end{lemma}
The first three proofs immediately follow from the definition of $\bpre$, the fact that $\bsim$ is a preorder and context composition; the last two are straightforward coinductive arguments.


\section{Context Lemma}

The objective of CCFB3 is to prove that making barbed bisimilarity
sensitive to substitutions and parallel composition is sufficient to
establish barbed congruence. This result is an instance of a \emph{context}
lemma, namely a characterization of congruence that relaxes the
universal quantification on arbitrary contexts, thereby simplifying
proofs of congruence.


We recall that \textit{substitutions} are endofunctions on names with finite support; they can be extended to processes and actions by simultaneously replacing all their free occurrences of names. To encode them, we rely on Beluga's built-in notion of simultaneous substitutions. In particular, given two Beluga contexts {\ttfamily \small g},{\ttfamily \small h:ctx}, a substitution {\ttfamily \small \$S:\$[h |- g]} maps the names in {\ttfamily \small g} to names in {\ttfamily \small h}; given a process \mbox{\ttfamily \small P:[g |- proc]} that depends on names in {\ttfamily \small g},  the process {\ttfamily \small P[\$S]:[h |- proc]} is obtained by replacing its names according to {\ttfamily \small \$S}.

Next, we denote as
$\bpree \coloneqq \{ (P,Q) \mid \text{for all} \ R, \sigma, (P \sigma \mid R)
\bsim (Q \sigma \mid R) \}$ the relation obtained by closing barbed
similarity under parallel composition and 
substitutions. Analogously to barbed precongruence, it is formalized by an inductive type {\ttfamily \small \color{belugagreen} BarbPre'}.~\clink{1\_definitions.bel}{225}{227}



The strategy to prove the context lemma consists in showing that $\bpree$ is a precongruence included in strong barbed similarity; since, by Lemma~\ref{lem:larg_precong}, strong barbed precongruence is the largest precongruence included in strong barbed similarity, it follows that $\bpree \subseteq \bpre$. Below, we detail the results required to complete the proof.

First, an auxiliary lemma describing how transitions of a process $P\sigma$ arise from transitions of $P$:

\begin{lemma} \label{le:auxsubst}
 If $P \sigma \xrightarrow{\alpha} P'$ and $\alpha \neq \tau$, then there are $\beta, P''$ such that $P \xrightarrow{\beta} P''$ and $\alpha = \beta \sigma$.~\clink{9\_context\_lemma.bel}{3}{81}
\end{lemma}
\begin{proof}
    (Sketch) 
    By inversion on $\alpha$ and then structural induction on $P$. In the Beluga encoding, the three cases $\alpha = x(y), \bar{x}y$ and $\bar{x}(y)$ are addressed separately.
\end{proof}

Proving that $\bpree \subseteq \bsim$ is straightforward:

\begin{lemma} \label{le:inclbarbsim}
    $\bpree$ is included in strong barbed similarity.~\clink{9\_context\_lemma.bel}{86}{88}
\end{lemma}
\begin{proof}
    It is sufficient to observe that, if $P \bpree Q$, then $P \equiv P (${\myfont Id}$) \mid 0 \ \bsim \ Q (${\myfont Id}$) \mid 0 \equiv Q$, where {\myfont Id} denotes the identity substitution.
\end{proof}

The proof that $\bpree$ is a precongruence is  trickier and is based on the following result --- this is, in fact, the only instance in which coinductive up-to techniques are used.

\begin{lemma} \label{le:auxbarbsim}
    The following relations are included in strong barbed similarity:
    \begin{enumerate}
        \item $\mathcal{R}_1\coloneqq  \{ (x(y).(P \sigma) \mid R, \ x(y).(Q \sigma) \mid R) \colon P \bpree Q\}$.~\clink{9\_context\_lemma.bel}{91}{114}
        \item $\mathcal{R}_2\coloneqq \{ (\overline{x}y.(P \sigma) \mid R, \ \overline{x}y.(Q \sigma) \mid R) \colon P \bpree Q\}$.~\clink{9\_context\_lemma.bel}{117}{140}
        \item $\mathcal{R}_3\coloneqq\{ (!(P \sigma) \mid R, \ !(Q \sigma) \mid R) \colon P \bpree Q\}$.~\clink{9\_context\_lemma.bel}{143}{195}
    \end{enumerate}
\end{lemma}
\begin{proof}
  We prove that $\mathcal{R}_1 \cup \bsim$ and $\mathcal{R}_2 \cup \bsim$ are strong barbed simulations, while $\mathcal{R}_3$ is a strong barbed simulation up to $\bsim$. The proofs, both as pen-and-paper and in Beluga, are carried out by coinduction, considering each possible observation and constructing the corresponding object required by the definition.
    \begin{enumerate}
        \item Checking that the elements of $\bsim$ satisfy all the observations is immediate, hence we focus on the elements of $\mathcal{R}_1$.
        Proving that $\mathcal{R}_1 \cup \bsim$ is barb preserving is also immediate, by looking at the structure of the pairs of processes in $\mathcal{R}_1$. To show that it is a reduction simulation, we proceed by inversion on the internal transition $s$ that needs to be matched; in particular, we illustrate the case in which $s$ has been derived from the {\small \textsc{S-COM-R}} rule, since it highlights one of the properties of substitutions that Beluga directly supports.
        
The transition $s$ has the form $x(z).P\sigma \mid R \xrightarrow{\tau} (P\sigma)\{y/z\} \mid R'$
        for some $y,R'$; similarly, we can see that
        $x(z).Q\sigma \mid R \xrightarrow{\tau} (Q\sigma)\{y/z\} \mid R'$. 
        As the composition of substitutions is a substitution, we conclude by observing that
        $ ((P\sigma)\{y/z\} \mid R', \ (Q\sigma)\{y/z\} \mid R') = (P (\sigma \circ \{y/z\}) \mid R', \ Q (\sigma \circ \{y/z\}) \mid R') \in \mathcal{R}_1$.
        \item Analogous to the previous case.
        \item Proving that $\mathcal{R}_3$ is barb preserving requires showing that a transition $s: \ !(P \sigma) \xrightarrow{\alpha} P'$, with $\alpha \neq \tau$, is matched by $!(Q \sigma)$. After inversion on $s$ through the {\small \textsc{S-REP}} rule, the desired transition is built after applying Lemma~\ref{le:auxsubst}.
        
        Conversely, proving the reduction closure up to $\bsim$ of $\mathcal{R}_3$
        does not require any explicit inversion; rather, it is enough to apply Lemma~\ref{lts:two} and the coinductive hypothesis to build a long chain of structural congruences and strong barbed similarities.
    \end{enumerate}
\end{proof}

\begin{lemma}\label{le:isprecong}
    $\bpree$ is a precongruence.~\clink{9\_context\_lemma.bel}{198}{220}
\end{lemma}
\begin{proof}
    By induction on the structure of the given context. The prefixes and replication cases make use of Lemma~\ref{le:auxbarbsim}, while the others rely on chains of $\equiv$ and $\bsim$.
\end{proof}

\begin{lemma}
    $\bpree$ is included in the largest precongruence included in strong barbed similarity.~\clink{9\_context\_lemma.bel}{223}{226}
\end{lemma}
\begin{proof}
    Follows immediately from Lemmas~\ref{le:inclbarbsim} and~\ref{le:isprecong}.
\end{proof}

\begin{theorem}[Context Lemma]
    $\bpree$ is included in strong barbed precongruence.~\clink{9\_context\_lemma.bel}{229}{231}
\end{theorem}
\begin{proof}
    Follows from Lemma~\ref{lem:larg_precong}, characterizing barbed precongruence as the largest precongruence included in strong barbed similarity.
\end{proof}

\section{Evaluation}\label{sec:eval}

The entire development comprises approximately 1500 lines of code,
organized into nine files. It includes 23 definitions --- three of
which are coinductive --- accounting for roughly $10\%$ of the total
codebase, as well as 53 theorems. Most of these theorems exactly correspond to
the relevant results presented in~\cite{SWBook}; however, several of the more
intricate ones must be decomposed into multiple statements. This is
particularly necessary when they involve combinations of induction and
coinduction with differing proof goals. Notably, the only
strengthening lemmas required in the Beluga formalization are those
already needed at the paper level to satisfy the side conditions of
the LTS rules concerning free and bound names. This once again
underscores the advantages of using HOAS for handling binders.

A central feature of the $\pi$-calculus is its ability to pass names, enabling subprocesses to communicate channel names and subsequently use them. Concretely, input prefixes bind received names, and in all standard semantics there is a transition in which the binder is eliminated and the bound name is replaced by a fresh free name representing the received value. In Beluga, such single substitutions are treated as a special case of simultaneous substitutions, which allows them to be composed seamlessly, as demonstrated in Lemma~\ref{le:auxbarbsim}. 
This is a net gain compared to other approaches (such as~\cite{Bengtson2009}), where the two notions need to be implemented from scratch.

One complication arising from substitutions in Beluga is their
interaction with the coverage checker. Verifying coverage often
requires solving non-trivial unification problems, which can become
more challenging in the presence of substitutions: this is well understood in the theory, but simply not supported yet by the Beluga implementation. Consequently,
Lemma~\ref{le:auxsubst} is formalized without a totality check.
In fairness, we find more urgent the 
  implementation of the  \emph{productivity} checker: for what it is worth, we have
  manually checked that all coinductive calls comply with
  productivity, as we understand it, but this falls somewhat short of a complete assurance
  of the correctness of our mechanization.


In the CCFB paper, the authors stated:

\begin{quote}
\emph{The crux of our challenge is the effective use of coinductive up-to
  techniques.  The intention is that the result should be relatively
  easy to achieve once the main properties of bisimilarity are
  established.
}\end{quote}


However, coinductive up-to techniques, which have a prominent role in
the study of other bisimilarities, are only used here once, i.e.\ to
prove the context lemma, and the formalization of strong barbed
similarity up to $\bsim$ to does not diverge significantly from the
formalization of strong barbed similarity. We conclude that the design of the challenge does not exercise significantly the up-to aspect.


\section{Related Work}
\label{sec:rel}


In this section we mostly concentrate on mechanizations of behavioral equivalence in process calculi.

In~\cite{Bengtson2009,BengtsonParrow2007Completeness}, Bengtson and
Parrow provide the largest formalization of the $\pi$-calculus in the
literature using nominal logic in Isabelle/HOL to represent binders.  The
encoding of the operational semantics uses
\emph{commitments}~\cite{Milner1993} rather than labelled transitions, i.e.\
pairs of an action and a derivative process that binds the bound names
of the action in the derivative. The treatment of replication
follows the single-rule presentation of~\cite{DBLP:journals/tcs/HonsellMS01}.  Their results include
proving that strong and weak bisimulations, both early and late, are
congruences, and that structural congruence is a
bisimulation; 
the authors also provide an
axiomatization of strong late bisimulation and prove that it is sound
and complete.
Coinductive definitions and proofs are encoded by using Isabelle/HOL's
standard coinductive package. Since most of these details are handled
by Isabelle, the main focus of Bengtson and Parrow is on the nominal
logic aspects of the formalization, which, while elegantly supporting
free and bound names, still entails a fair amount of work to establish
basic infrastructural properties.

Ambal et al.~\cite{AmbalLS21} provide a Coq formalization of the
higher-order $\pi$-calculus that systematically studies binder
representations for mixed binding (process vs name binders) and
establishes strong context bisimilarity as a congruence via Howe’s
method.  Coinduction is used in the standard way to characterize
bisimilarity; coinductive ``guarded'' proofs only appear in the final
step showing that the Howe closure is included in
bisimilarity. A direct Coq development in the style of ours would likely require libraries such as PACO~\cite{Hur2013}. Further, the HOAS approach encodes process and name binders in the same way, offering significant simplifications.
In~\cite{TianS20}, Tian and Sangiorgi study proof methods for bisimulation in CCS, focusing on contractions rather than equations. They formalize the congruence induced by weak bisimilarity (rooted bisimilarity) in HOL4 using the coinductive package \texttt{Hol\_coreln}.
Their development does not treat processes up to $\alpha$-conversion, and restriction is not a binder; contexts are therefore represented as single-binder $\lambda$-expressions applied to processes.  As in our work, the coarsest congruence contained in bisimilarity is characterized via closure under composition, although here it is ancillary to the main result on unique solutions of rooted contractions.

Tiu and Miller~\cite{DBLP:journals/tocl/TiuM10} give a proof-search
specification of the finite $\pi$-calculus in a logic equipped with
definitions, fixed points, and the $\nabla$ quantifier for generic
judgments. Their encoding is equivalent to ours w.r.t.\ syntax and
LTS. It does seem more general insofar that the interaction between
the $\forall$, $\exists$, and $\nabla$ quantifiers explains the usual
distinctions between early, late, and, remarkably, open
bisimulation. For the latter  the $\nabla$ quantifier plays a
role analogous to Sangiorgi's
\emph{distinctions}. 

Veltri and Vezzosi~\cite{VeltriV23} give a much more extensive
formalization of process-calculus semantics in Guarded Cubical Agda:
they develop fully abstract denotational models for CCS and the early
$\pi$-calculus, showing that denotational equality coincides with
bisimilarity. Their encoding uses well-scoped de Bruijn syntax, in
contrast with our HOAS representation in Beluga. Their approach to
coinduction is also quite different: guarded recursion, clocks, and
cubical path equality provide the semantic infrastructure through
which bisimilarity is recovered as equality, whereas we reason
directly with coinductive operational relations.

\section{Conclusions and Future Work}\label{sec:conc}

What happens if  a challenge turns out to be not too challenging? Is it the designers' fault or the encoders' merit? 
In our previous paper we mused on how uneventful the 
formalization had been. We have to repeat ourselves: the HOAS encoding
of the syntax and semantics of the calculus, which we have inherited
from a long tradition (dating back to~\cite{MillerPI}), abstracts all
the issues related to free and bound names that have preoccupied all 
first-order encodings. What the present paper brings to the table is a
very intensive use of coinductive reasoning: here, Beluga's copattern
approach shines, much more than in a previous coinductive case study
(Howe's method, see~\cite{MomiglianoPT19}, which turns out to
exercise coinduction only in a limited fashion).

Beluga's underlying sized types discipline is intrinsically
compositional, making nesting coinductive proofs completely
unproblematic. This is in sharp contrast with
the guarded recursion approach native to Rocq, which would have been
simply unfeasible and compelled the user to switch to dedicated
libraries such as PACO~\cite{Hur2013}. Again, it is the support for
both HOAS and coinduction via observation that makes, in our judgment,
the development so smooth --- Agda supports a similar approach to coinduction, but not HOAS; Abella offers the converse trade-off.

\paragraph{Future Work.}

There are several directions in which the present formalization could be  extended.
A natural first step is to establish the converse direction of the
context
lemma. 
The proof presented  in~\cite{SWBook}
proceeds 
by constructing a suitable context such that a pair of strongly barbed
congruent processes can be composed in parallel under substitution so as to yield
strongly barbed bisimilar processes. The argument relies on 
multi-step internal transitions whose length depends on the size of the
support of the substitution. 
It is not clear to us whether Beluga's theory of substitution 
directly captures this style of argument, particularly when it involves
a fine-grained analysis of substitutions and their supports. Moreover,
the proof makes essential use of the symmetry of strong barbed
bisimilarity, so we would have to carry out the symmetric proofs after all.


\emph{Open} bisimulation is particularly appealing in our setting
because Beluga’s contextual objects and simultaneous substitutions
seem to provide much of the infrastructure needed to express
substitution-closed behavioral relations. The main additional burden
would be the treatment of distinctions/freshness constraints and their
interaction with the LTS. Since Beluga does not have a $\nabla$
quantifier, one must prove that the resulting substitution-based
relation coincides with the standard distinction-indexed definition,
or else identify precisely which variant of open similarity has been
mechanized. One possibility to avoid an explicit 
representation of distinctions could be to look at the theory of \emph{quasi-open} bisimilarity~\cite{Fu05,HorneALT18}.

Another technically involved extension concerns the coincidence
between strong barbed congruence and early congruence. The standard
proof requires the introduction of a stratification of strong early
bisimilarity, with the key inclusion shown via a contrapositive
argument. This reliance on classical reasoning renders a direct
mechanization in Beluga  particularly
challenging.


Finally, Lassen in~\cite{Lassen98} studies contextual closure as the
greatest \emph{adequate} congruence in the context of program
equivalence for a functional language. Since adequacy has a barbed
``flavor'', the question is whether analogous characterizations can
be established for bisimilarity-based congruences in the
$\pi$-calculus.

\section*{Acknowledgments}
\begin{sloppypar}
  This work is partially supported by the National Science Foundation under \href{https://www.nsf.gov/awardsearch/showAward?AWD_ID=2242786}{Grant No.\ 2242786 (SHF:Small:Concurrency In Reversible Computations)}.	We thank Brigitte Pientka for several helpful conversations.
\end{sloppypar}
\smallskip
\noindent
\emph{AI usage statement}. The formalization, definitions, and proofs presented in this paper are entirely the authors’ work. Generative AI tools were used solely for editorial assistance (improving clarity and English wording), and had no role in the technical development or validation of the results.
	
\bibliographystyle{eptcs}
\bibliography{bib/merged}

\appendix

\end{document}